\documentclass{article}

\usepackage{times}
\usepackage{graphicx} 
\usepackage{subfigure}

\usepackage{natbib}

\usepackage{algorithm}
\usepackage{algorithmic}
\usepackage[algo2e,linesnumbered,lined,boxed,vlined]{algorithm2e}

\usepackage{amsfonts}
\usepackage{amssymb}

\usepackage{hyperref}

\usepackage[accepted]{icml2016}

\usepackage{comment}
\usepackage{amsthm}
\usepackage{mathtools}
\usepackage{tabularx}
\usepackage{multirow}

\newtheorem{lemma}{Lemma}
\newtheorem{proposition}{Proposition}
\newtheorem{corollary}{Corollary}
\newtheorem{definition}{Definition}
\newtheorem{example}{Example}
\newtheorem{remark}{Remark}

\def\b{\ensuremath\boldsymbol}

\newcommand{\indep}{\perp\!\!\!\perp}

\usepackage{lastpage}
\usepackage{fancyhdr}
\usepackage{url}
\icmltitlerunning{Sampling Algorithms, from Survey Sampling to Monte Carlo Methods: Tutorial and Literature Review}

\begin{document}

\twocolumn[
\icmltitle{Sampling Algorithms, from Survey Sampling to Monte Carlo Methods: Tutorial and Literature Review}

\icmlauthor{Benyamin Ghojogh*}{bghojogh@uwaterloo.ca}
\icmladdress{Department of Electrical and Computer Engineering, 
\\Machine Learning Laboratory, University of Waterloo, Waterloo, ON, Canada}
\icmlauthor{Hadi Nekoei*}{hnekoeiq@uwaterloo.ca}
\icmladdress{MILA (Montreal Institute for Learning Algorithms) -- Quebec AI Institute, Montreal, Quebec, Canada}
\icmlauthor{Aydin Ghojogh*}{aydin.ghojogh@gmail.com}
\icmladdress{ }
\icmlauthor{Fakhri Karray}{karray@uwaterloo.ca}
\icmladdress{Department of Electrical and Computer Engineering, 
\\Centre for Pattern Analysis and Machine Intelligence, University of Waterloo, Waterloo, ON, Canada}
\icmlauthor{Mark Crowley}{mcrowley@uwaterloo.ca}
\icmladdress{Department of Electrical and Computer Engineering, 
\\Machine Learning Laboratory, University of Waterloo, Waterloo, ON, Canada}

* The first three authors contributed equally to this work.

\icmlkeywords{Tutorial}

\vskip 0.3in
]

\begin{abstract}
This paper is a tutorial and literature review on sampling algorithms. We have two main types of sampling in statistics. The first type is survey sampling which draws samples from a set or population. The second type is sampling from probability distribution where we have a probability density or mass function. In this paper, we cover both types of sampling. First, we review some required background on mean squared error, variance, bias, maximum likelihood estimation, Bernoulli, Binomial, and Hypergeometric distributions, the Horvitz–Thompson estimator, and the Markov property. Then, we explain the theory of simple random sampling, bootstrapping, stratified sampling, and cluster sampling. We also briefly introduce multistage sampling, network sampling, and snowball sampling. Afterwards, we switch to sampling from distribution. We explain sampling from cumulative distribution function, Monte Carlo approximation, simple Monte Carlo methods, and Markov Chain Monte Carlo (MCMC) methods. For simple Monte Carlo methods, whose iterations are independent, we cover importance sampling and rejection sampling. For MCMC methods, we cover Metropolis algorithm, Metropolis-Hastings algorithm, Gibbs sampling, and slice sampling. Then, we explain the random walk behaviour of Monte Carlo methods and more efficient Monte Carlo methods, including Hamiltonian (or hybrid) Monte Carlo, Adler's overrelaxation, and ordered overrelaxation. Finally, we summarize the characteristics, pros, and cons of sampling methods compared to each other. This paper can be useful for different fields of statistics, machine learning, reinforcement learning, and computational physics. 
\end{abstract}

\section{Introduction}\label{section_introduction}

Sampling is a fundamental task in statistics. However, this terminology is used for two different tasks in statistics. On one hand, sampling refers to \textit{survey sampling} which is selecting instances from a population or set:
\begin{align}
\mathcal{D} := \{x_1, x_2, \dots, x_N\},
\end{align}
where the population size is $N := |\mathcal{D}|$. 
Note that some of the instances of this population may be repetitive numbers/vectors. Survey sampling draws $n$ samples from the population $\mathcal{D}$ to have a set of samples $\mathcal{S}$ where $n := |\mathcal{S}|$. 
There are several articles and books on survey sampling such as \cite{barnett1974elements,smith1976foundations,foreman1991survey,schofield1996survey,nassiuma2001survey,chaudhuri2005survey,tille2006sampling,mukhopadhyay2008theory,scheaffer2011elementary,fuller2011sampling,tille2012survey,hibberts2012common,singh2013elements,kalton2020introduction}. 
It is a field of research in statistics, with many possible future developments \cite{brick2011future}, especially in distributed networks and graphs \cite{frank2011survey,heckathorn2017network}. 
Some of the popular methods in survey sampling are Simple Random Sampling (SRS) \cite{barnett1974elements}, bootstrapping \cite{efron1994introduction}, stratified sampling, cluster sampling \cite{barnett1974elements}, multistage sampling \cite{lance2016sampling}, network sampling \cite{schofield2011network}, and snowball sampling \cite{goodman1961snowball}. 

On the other hand, sampling can refer to drawing samples from probability distributions. 
Usually, in real-world applications, distributions of data are complicated to sample from; for example, they can be mixture of several distributions \cite{ghojogh2019fitting}. One can approximate samples from the complicated distributions by sampling from some other simple-to-sample distribution. 
The sampling methods which perform this sampling approximation are referred to as the Monte Carlo methods \cite{mackay1998introduction,bishop2006pattern,kalos2009monte,hammersley2013monte,kroese2013handbook}. Monte Carlo approximation \cite{kalos2009monte} can be used for estimating the expectation or probability of a function of data over the data distribution. Monte Carlo methods can be divided into two main categories, i.e., simple methods and Markov Chain Monte Carlo (MCMC) \cite{mackay2003information}.
Note that Monte Carlo methods are iterative.
In simple Monte Carlo methods, every iteration is independent from previous iterations and drawing samples is performed blindly. Importance sampling \cite{glynn1989importance} and rejection sampling \cite{casella2004generalized,bishop2006pattern,robert2013monte} are examples of simple Monte Carlo methods. 
In MCMC \cite{murray2007advances}, however, every iteration is dependent on its previous iteration because they have the memory of Markov property \cite{koller2009probabilistic}. 
Some examples of MCMC are Metropolis algorithm \cite{metropolis1953equation}, Metropolis-Hastings algorithm \cite{hastings1970monte}, Gibbs sampling \cite{geman1984stochastic}, and slice sampling \cite{neal2003slice,skilling2003slice}. The Metropolis algorithms are usually slow because of their random walk behaviour \cite{spitzer2013principles}. Some efficient methods, for faster exploration of range of data by sampling methods, are Hamiltonian (or hybrid) Monte Carlo method \cite{duane1987hybrid}, Adler's overrelaxation \cite{adler1981over}, and ordered overrelaxation \cite{neal1998suppressing}. 
Monte Carlo methods have been originally developed in computational physics \cite{newman2013computational}; hence, they have application in physics \cite{binder2012monte}. They also have application in other fields such as finance \cite{glasserman2013monte} and reinforcement learning \cite{barto1994monte,wang2012monte,sutton2018reinforcement}. 

In this tutorial and literature review paper, we cover both areas of sampling, i.e., survey sampling and sampling from distributions using Monte Carlo methods. The remainder of this paper is organized as follows. Section \ref{section_background} reviews some required background on mean squared error, variance, bias, estimations using maximum likelihood estimation, Bernoulli, Binomial, and Hypergeometric distributions, the Horvitz–Thompson estimator, and the Markov property. We introduce, in detail, the methods of survey sampling and Monte Carlo methods in Sections \ref{section_survey_sampling} and \ref{section_Monte_Carlo}, respectively. Finally, we provide a summary of methods, their pros and cons, and conclusions in Section \ref{section_conclusion}.

\section{Background}\label{section_background}

\subsection{Mean Squared Error, Variance, and Bias}

The materials of this subsection are taken from our previous tutorial paper \cite{ghojogh2019theory}. 
Assume we have variable $X$ and we estimate it. Let the random variable $\widehat{X}$ denote the estimate of $X$. Let $\mathbb{E}(\cdot)$ and $\mathbb{P}(\cdot)$ denote expectation and probability, respectively. 
The \textit{variance} of estimating this random variable is defined as:
\begin{align}\label{equation_variance}
\mathbb{V}\text{ar}(\widehat{X}) := \mathbb{E}\big((\widehat{X} - \mathbb{E}(\widehat{X}))^2\big),
\end{align}
which means average deviation of $\widehat{X}$ from the mean of our estimate, $\mathbb{E}(\widehat{X})$, where the deviation is squared for symmetry of difference.
This variance can be restated as:
\begin{align}
\mathbb{V}\text{ar}(\widehat{X}) &= \mathbb{E}(\widehat{X}^2) - (\mathbb{E}(\widehat{X}))^2. \label{equation_variance_2}
\end{align}
See Appendix \ref{section_appendix_background} for proof.

Our estimation can have a bias. The \textit{bias} of our estimate is defined as:
\begin{align}\label{equation_bias}
\mathbb{B}\text{ias}(\widehat{X}) := \mathbb{E}(\widehat{X}) - X,
\end{align}
which means how much the mean of our estimate deviates from the original $X$.

\begin{definition}[Unbiased Estimator]\label{definition_unbiased_estimator}
If the bias of an estimator is zero, i.e., $\mathbb{E}(\widehat{X}) = X$, the estimator is \textbf{unbiased}. 
\end{definition}

The \textit{Mean Squared Error (MSE)} of our estimate, $\widehat{X}$, is defined as:
\begin{align}\label{equation_MSE}
\text{MSE}(\widehat{X}) := \mathbb{E}\big((\widehat{X} - X)^2\big),
\end{align}
which means how much our estimate deviates from the original $X$.


The relation of MSE, variance, and bias is as follows:
\begin{align}
&\text{MSE}(\widehat{X}) = \mathbb{V}\text{ar}(\widehat{X}) + (\mathbb{B}\text{ias}(\widehat{X}))^2. \label{equation_relation_MSE_variance_bias}
\end{align}
See Appendix \ref{section_appendix_background} for proof.

If we have two random variables $\widehat{X}$ and $\widehat{Y}$, we can say:
\begin{align}
&\mathbb{V}\text{ar}(a\widehat{X} + b\widehat{Y}) \nonumber \\
&~~~~~~~~= a^2\, \mathbb{V}\text{ar}(\widehat{X}) + b^2\, \mathbb{V}\text{ar}(\widehat{X}) + 2ab\, \mathbb{C}\text{ov}(\widehat{X},\widehat{Y}), \label{equation_variance_of_two_variables}
\end{align}
where $\mathbb{C}\text{ov}(\widehat{X},\widehat{Y})$ is \textit{covariance} defined as:
\begin{align}\label{equation_covariance}
\mathbb{C}\text{ov}(\widehat{X},\widehat{Y}) := \mathbb{E}(\widehat{X}\widehat{Y}) - \mathbb{E}(\widehat{Y})\,\mathbb{E}(\widehat{Y}).
\end{align}
See Appendix \ref{section_appendix_background} for proof.

If the two random variables are independent, i.e., $X \indep Y$, we have:
\begin{align}
&\mathbb{E}(\widehat{X}\widehat{Y}) = \mathbb{E}(\widehat{X})\, \mathbb{E}(\widehat{Y}) \implies \mathbb{C}\text{ov}(\widehat{X},\widehat{Y}) = 0, \label{equation_expectation_independent}
\end{align}
See Appendix \ref{section_appendix_background} for proof.
Note that Eq. (\ref{equation_expectation_independent}) is not true for the reverse implication (we can prove by counterexample). 

We can extend Eqs. (\ref{equation_variance_of_two_variables}) and (\ref{equation_covariance}) to multiple random variables:
\begin{align}
&\mathbb{V}\text{ar}\Big(\sum_{i=1}^k a_i X_i \Big) \nonumber \\
&~~~~~~~ = \sum_{i=1}^k a_i^2\, \mathbb{V}\text{ar}(X_i) + \sum_{i=1}^k \sum_{j=1, j\neq i}^k a_i a_j \mathbb{C}\text{ov}(X_i, X_j), \label{equation_variance_multiple} \\
&\mathbb{C}\text{ov}\Big( \sum_{i=1}^{k_1} a_i X_i, \sum_{j=1}^{k_2} b_j Y_j \Big) = \sum_{i=1}^{k_1} \sum_{j=1}^{k_2} a_i\, b_j\, \mathbb{C}\text{ov}(X_i, Y_j),
\end{align}
where $a_i$'s and $b_j$'s are not random.

According to Eq. (\ref{equation_expectation_independent}), if the random variables are independent, Eq. (\ref{equation_variance_multiple}) is simplified to:
\begin{align}
&\mathbb{V}\text{ar}\Big(\sum_{i=1}^k a_i X_i \Big) = \sum_{i=1}^k a_i^2\, \mathbb{V}\text{ar}(X_i). \label{equation_variance_multiple_independent}
\end{align}

\subsection{Estimates for Mean and Variance}

The Maximum Likelihood Estimation (MLE) or Method of Moments (MOM) for a the mean and variance of Gaussian distributions are:
\begin{align}
& \mu = \frac{1}{N} \sum_{j=1}^N x_j, \label{equation_mean_estimate}\\
& \sigma^2 = \frac{1}{N} \sum_{j=1}^N (x_j - \mu)^2, \label{equation_variance_estimate}
\end{align}
respectively. These estimates are usually used for estimating the mean and variance of any data. 

\begin{lemma}\label{lemma_variance_estimate_restate}
The estimate of variance, which is Eq. (\ref{equation_variance}), can also be restated as:
\begin{align}\label{equation_variance_estimate_2}
\sigma^2 = \frac{1}{N} \sum_{j=1}^N x_j^2 - \mu^2.
\end{align}
\end{lemma}
\begin{proof}
See Appendix \ref{section_appendix_background} for proof.
\end{proof}

\begin{lemma}
The variance of the estimate of mean is:
\begin{align}\label{equation_variance_of_mean}
\mathbb{V}\text{ar}(\mu) = \frac{1}{N}\, \sigma^2.
\end{align}
\end{lemma}
\begin{proof}
See Appendix \ref{section_appendix_background} for proof.
\end{proof}

\begin{proposition}\label{proposition_variance_unbiased}
An \textbf{unbiased} estimator for variance is:
\begin{align}\label{equation_variance_unbiased}
\sigma^2 = \frac{1}{N-1} \sum_{j=1}^N (x_j - \mu)^2.
\end{align}
\end{proposition}
\begin{proof}
See Appendix \ref{section_appendix_background} for proof.
\end{proof}

Note that Eq. (\ref{equation_variance_estimate}) is a biased estimate of variance because its expectation is:
\begin{align*}
\mathbb{E}(\sigma^2) &= \frac{1}{N} \mathbb{E}\Big(\sum_{j=1}^N (x_j - \mu)^2\Big) = \frac{1}{N} (N-1)\, \sigma^2.
\end{align*}

\subsection{Bernoulli, Binomial, and Hypergeometric Distributions}

Bernoulli distribution is a discrete distribution of being one and zero with probabilities $p$ and $1-p$, respectively. Its expected value and variance are:
\begin{align}
&\mathbb{E}(X) = p, \label{equation_Bernoulli_mean}\\
&\mathbb{V}\text{ar}(X) = p\,(1-p), \label{equation_Bernoulli_variance}
\end{align}
respectively. 

Binomial distribution is a discrete distribution for probability of success of $n$ \textit{independent} events out of $N$ events where the probability of success of every event is $p$. 
As the drawn events are independent, binomial distribution can be seen like sampling \textit{with replacement}. 
The Probability Mass Function (PMF) of binomial distribution is:
\begin{align}\label{equation_binomial_PMF}
f(n) = \binom{N}{n} p^{n} (1-p)^{N-n}.
\end{align}

Hypergeometric distribution is a discrete distribution for probability of $k$ successes in $n$ draws, \textit{without replacement}, out of $N$ events where $K$ success actually exist in the $N$ events. 
Binomial distribution can be seen like sampling \textit{without replacement}. 
The PMF of hypergeometric distribution is:
\begin{align}\label{equation_hypergeometric_PMF}
f(k) = \frac{\binom{K}{k} \binom{N-K}{n-k}}{\binom{N}{n}}.
\end{align}

\subsection{The Horvitz–Thompson Estimator}

Consider the following estimator for the population quantity from a sample of size $N$:
\begin{align}\label{equation_estimator_population_quantity}
\theta = \sum_{j=1}^N h(x_j).
\end{align}
Some special cases of this estimator are:
\begin{align}
\text{total (sum): } &h(x_j) = x_j, \\
\text{mean (average): } &h(x_j) = \frac{x_j}{N}, \label{equation_population_quantity_mean} \\
\text{proportion (of set $\mathcal{S}$): } &h(x_j) = \frac{\mathbb{I}(x_j \in \mathcal{S})}{N},
\end{align}
where $\mathbb{I}(x_j \in \mathcal{S}) = \mathbb{I}_j$ denotes the indicator function which is zero or one if $x_j$ does not belong or belongs to the set $\mathcal{S}$, respectively. 

However, there is an estimator named the Horvitz–Thompson (HT) estimator \cite{horvitz1952generalization}, defined as:
\begin{align}\label{equation_HT_estimator}
\widehat{\theta}_\text{HT} := \sum_{j \in \mathcal{S}} \frac{h(x_j)}{\pi_j},
\end{align}
where $\pi_j := \mathbb{P}(j \in \mathcal{S})$. 
The HT estimator can also be used to estimate the population quantity of data \cite{little2019statistical}. 

\begin{definition}[Inverse Probability Weighting]
In inverse probability weighting, if an individual has large/small probability of being included, we deflate/inflate its value. This technique, which is common for deriving estimators, reduces the bias of unweighted estimator \cite{robins1994estimation}. 
\end{definition}

The Eq. (\ref{equation_HT_estimator}) shows that the HT estimator uses the inverse probability weighting because of having probability of inclusion, $\pi_j$, in the denominator. Hence, its bias is reduced; actually it is unbiased. 

\begin{proposition}\label{proposition_HT_unbiased}
The HT estimator is an \textbf{unbiased} estimator for the population quantity. 
\end{proposition}
\begin{proof}
See Appendix \ref{section_appendix_background} for proof.
\end{proof}

\subsection{The Markov Property and Markov Chain}\label{section_Markov_property}

This subsection is taken from our previous tutorial on hidden Markov model \cite{ghojogh2019hidden}. 
Consider a times series of random variables $X_1, X_2, \dots, X_n$.
In general, the joint probability of these random variables can be written as:
\begin{align}
\mathbb{P}(X_1, &X_2, \dots, X_n) = \mathbb{P}(X_1)\, \mathbb{P}(X_2 \,|\, X_1)\, \nonumber \\
&\mathbb{P}(X_3 \,|\, X_2, X_1) \dots \mathbb{P}(X_n \,|\, X_{n-1}, \dots, X_2, X_1),
\end{align}
according to chain (or multiplication) rule in probability. 
[The first order] \textit{Markov property} is an assumption which states that in a time series of random variables $X_1, X_2, \dots, X_n$, every random variable is merely dependent on the latest previous random variable and not the others. In other words:
\begin{align}
\mathbb{P}(X_i \,|\, X_{i-1}, X_{i-2}, \dots, X_{2}, X_{1}) = \mathbb{P}(X_i \,|\, X_{i-1}).
\end{align}
Hence, with Markov property, the chain rule is simplied to:
\begin{align}
\mathbb{P}&(X_1, X_2, \dots, X_n) \nonumber \\
&= \mathbb{P}(X_1)\, \mathbb{P}(X_2 \,|\, X_1)\, \mathbb{P}(X_3 \,|\, X_2) \dots \mathbb{P}(X_n \,|\, X_{n-1}).
\end{align}
The Markov property can be of any order. For example, in a second order Markov property, a random variable is dependent on the latest and one-to-latest variables. Usually, the default Markov property is of order one. 
A stochastic process which has the Markov process is called a Markovian process (or Markov process). 

A \textit{Markov chain} is a probabilistic graphical model \cite{koller2009probabilistic} which has Markov property. The Markov chain can be either directed or undirected. Usually, Markov chain is a Bayesian network where the edges are directed. It is important not to confuse Markov chain with Markov network. 
For more information on Markov property and Markov chains, refer to \cite{ghojogh2019hidden}. 

\begin{figure*}[!t]
\centering
\includegraphics[width=5in]{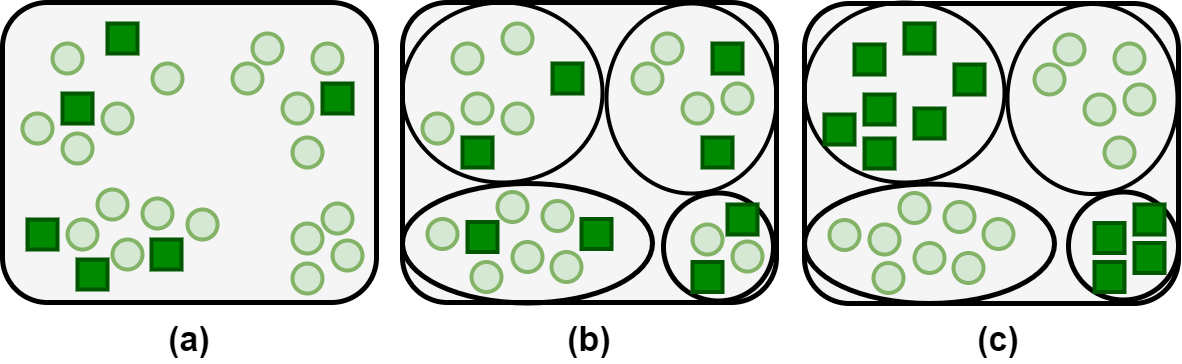}
\caption{An example of survey sampling where the bold squares denote the drawn samples and the circle sets show the strata or clusters of data: (a) SRS, (b) stratified sampling, and (c) cluster sampling.}
\label{figure_survey_sampling}
\end{figure*}

\section{Survey Sampling}\label{section_survey_sampling}

\subsection{Simple Random Sampling}

\begin{definition}[Simple Random Sampling]\label{definition_SRS}
Simple Random Sampling (SRS) is drawing a set $\mathcal{S}$ of size $n < N$ from the set of data $\mathcal{D}$ \textbf{without} replacement \cite{barnett1974elements}. In SRS, all items have the same probability of being chosen \cite{yates2002practice}. An illustration of SRS is shown in Fig. \ref{figure_survey_sampling}.
\end{definition}

\begin{corollary}
According to Definition \ref{definition_SRS}, there is no repetitive item in a sample drawn by SRS, if there is no repetitive item in the set $\mathcal{D}$. 
\end{corollary}

Using $h(x_j)$ for mean (see Eq. (\ref{equation_population_quantity_mean})) in the HT estimator (see Eq. (\ref{equation_HT_estimator})) can give us an estimator for the mean. According to the hypergeometric distribution, introduced Eq. (\ref{equation_hypergeometric_PMF}), we have in SRS (which is sampling without replacement): 
\begin{align}
\pi_j = \mathbb{P}(j \in \mathcal{S}) = \frac{\binom{1}{1} \binom{N-1}{n-1}}{\binom{N}{n}} = \frac{n}{N}.
\end{align}
Hence, the mean of sample is:
\begin{align}\label{equation_mean_SRS}
\widehat{\mu} = \sum_{j \in \mathcal{S}} \frac{x_j / N}{n / N} = \frac{1}{n} \sum_{j \in \mathcal{S}} x_j = \frac{1}{n} \sum_{j=1}^N x_j\, \mathbb{I}_j.
\end{align}
Compare this with the mean of whole data which is Eq. (\ref{equation_mean_estimate}). 
Similar to Eq. (\ref{equation_variance_unbiased}), the variance of sample in SRS is calculated as:
\begin{align}\label{equation_SRS_variance_estimate}
\widehat{\sigma}^2 = \frac{1}{n-1} \sum_{j \in \mathcal{S}} (x_j - \widehat{\mu})^2.
\end{align}

\begin{proposition}\label{proposition_expectation_variance_of_mean_SRS}
The expectation and variance of the mean of sample by SRS, i.e. Eq. (\ref{equation_mean_SRS}), are \cite{barnett1974elements}:
\begin{align}
&\mathbb{E}(\widehat{\mu}) = \mu, \label{equation_SRS_mean_expectation} \\
&\mathbb{V}\text{ar}(\widehat{\mu}) = \big(1 - \frac{n}{N}\big)\, \frac{\sigma^2}{n}, \label{equation_SRS_mean_variance}
\end{align}
respectively, where $\mu$ and $\sigma^2$ are the mean and variance of whole data, defined by Eqs. (\ref{equation_mean_estimate}) and (\ref{equation_variance_unbiased}), respectively.  
\end{proposition}
\begin{proof}
See Appendix \ref{section_appendix_survey_sampling} for proof.
\end{proof}

\begin{corollary}\label{corollary_SRS_varianceOfMean_proportional}
The variance of mean of sample by SRS can be restated as:
\begin{align}\label{equation_SRS_varianceOfMean_proportional}
\mathbb{V}\text{ar}(\widehat{\mu}) = \frac{1}{n} (1 - \frac{n}{N}) &\bigg[ \sum_{k=1}^K \frac{N_k - 1}{N - 1} \sigma_k^2 \nonumber \\
&+ \sum_{k=1}^K \frac{N_k}{N-1} (\mu_k - \mu)^2 \bigg].
\end{align}
\end{corollary}
\begin{proof}
See Appendix \ref{section_appendix_survey_sampling} for proof.
\end{proof}

\subsection{Bootstrapping}

\begin{definition}[Bootstrapping]
Bootstrapping \cite{efron1994introduction} is another name for simple random sampling \textbf{with} replacement \cite{pathak1962simple}. 
\end{definition}

Bootstrapping has been used in different statistical techniques such as inference \cite{mooney1993bootstrapping} and Bootstrap AGGregatING (bagging) \cite{breiman1996bagging} for model averaging \cite{hoeting1999bayesian,ghojogh2019theory}. It is noteworthy that, as bootstrapping is sampling with replacement, its probability distribution follows the binomial distribution introduced by Eq. (\ref{equation_binomial_PMF}). 

\subsection{Stratified Sampling}

\begin{definition}\label{definition_strata}
We define $\{\mathcal{D}_k\}_{k=1}^K$ to be strata (plural of stratum) of the dataset $\mathcal{D}$. 
The dataset is divided into $K$ disjoint strata:
\begin{align}\label{equation_disjoint_sets}
\mathcal{D} = \mathcal{D}_1 \oplus \mathcal{D}_2 \oplus  \cdots \oplus \mathcal{D}_K,
\end{align}
where:
\begin{align}
& \mathcal{D}_i \cap \mathcal{D}_j = \varnothing, ~~~ \forall i,j \in \{1, \dots, K\}, ~ i \neq j, \text{ and } \\
&\bigcup_{i=1}^K \mathcal{D}_i = \mathcal{D}.
\end{align}
As the strata are disjoint, they are independent. 
Note that strata may also be referred to as clusters or classes. 
\end{definition}

\begin{definition}[Stratified Sampling]\label{definition_stratified_sampling}
Let data $\mathcal{D}$ consist of $K$ strata defined in Eq. (\ref{equation_disjoint_sets}). 
Stratified sampling is simple random sampling (see Definition \ref{definition_SRS}) of size $n_k < N_k$ within every stratum $\mathcal{D}_k$, where $N_k := |\mathcal{D}_k|$ \cite{barnett1974elements}. 
An illustration of stratified sampling is shown in Fig. \ref{figure_survey_sampling}.
\end{definition}

\begin{corollary}\label{corollary_mean_variance_with_strata}
Suppose data consist of $K$ strata. According to Eqs. (\ref{equation_mean_estimate}) and (\ref{equation_variance_unbiased}), the actual mean and variance of the $k$-th stratum are:
\begin{align}
&\mu_k = \frac{1}{N_k} \sum_{j=1}^{N_k} x_{k,j}, \label{equation_mean_unbiased_stratum} \\
&\sigma_k^2 = \frac{1}{N_k-1} \sum_{j=1}^{N_k} (x_{k,j} - \mu_k)^2, \label{equation_variance_unbiased_stratum}
\end{align}
respectively. 
\end{corollary}

Using $h(x_j)$ for mean (see Eq. (\ref{equation_population_quantity_mean})) in the HT estimator (see Eq. (\ref{equation_HT_estimator})) can give us an estimator for the mean. 
Let the estimated mean within the $k$-th stratum be denoted by $\widehat{\mu}_k$. 
Moreover, let $\mathcal{D}_k = \{x_{k,1}, x_{k,2}, \dots, x_{k,N_k}\}$ and the sample by SRS within $\mathcal{D}_k$ be denoted by $\mathcal{S}_k$. In the HT estimator, we have $h(x_{k,j}) = x_{k,j} / N_k$ and $\mathbb{P}(x_{k,j} \in \mathcal{D}_k) = n_k / N_k$.
According to the HT estimator, we have:
\begin{align}
\widehat{\mu}_k &= \sum_{x_{k,j} \in \mathcal{S}_k} \frac{h(x_{k,j})}{\mathbb{P}(x_{k,j} \in \mathcal{Ss}_k)} = \sum_{x_{k,j} \in \mathcal{S}_k} \frac{x_{k,j} / N_k}{n_k / N_k} \nonumber \\
&= \sum_{x_{k,j} \in \mathcal{S}_k} \frac{x_{k,j}}{n_k},
\end{align}
which makes sense because $n_k$ instances are sampled by SRS from every stratum $\mathcal{S}_k$. 
Another way to obtain this is:
\begin{align}
\widehat{\mu}_k &= \frac{1}{n_k} \sum_{x_{k,j} \in \mathcal{S}_k} x_{k,j} = \frac{1}{n_k} \sum_{j=1}^{N_k} x_{k,j} \mathbb{I}(x_{k,j} \in \mathcal{S}_k). 
\end{align}

As we have SRS in every stratum, the estimate variance within the $k$-th stratum follows Eq. (\ref{equation_SRS_variance_estimate}) and is:
\begin{align}
\widehat{\sigma}_k^2 = \frac{1}{n_k-1} \sum_{j \in \mathcal{S}_k} (x_j - \widehat{\mu}_k)^2.
\end{align}

As we have SRS in every stratum, the expectation and variance of estimate of mean of the $k$-th stratum follow Eqs. (\ref{equation_SRS_mean_expectation}) and (\ref{equation_SRS_mean_variance}) as:
\begin{align}
&\mathbb{E}(\widehat{\mu}_k) = \mu_k, \label{equation_stratifiedSampling_mean_expectation}  \\
&\mathbb{V}\text{ar}(\widehat{\mu}_k) = \big(1 - \frac{n_k}{N_k}\big)\, \frac{\sigma_k^2}{n}, \label{equation_stratifiedSampling_mean_variance}
\end{align}
where $\mu_k$ and $\sigma_k$ are the actual mean and variance of the $k$-th stratum, respectively. 

\begin{lemma}\label{lemma_mean_variance_with_strata}
Suppose data consist of $K$ strata. The actual total mean and variance of data are:
\begin{align}
&\mu = \sum_{k=1}^K \frac{N_k}{N} \mu_k, \label{equation_mean_with_strata} \\
&\sigma^2 = \frac{1}{N-1} \Big[ \sum_{k=1}^K (N_k - 1) \sigma_k^2 + \sum_{k=1}^K N_k (\mu_k - \mu)^2 \Big], \label{equation_variance_with_strata}
\end{align}
respectively. 
\end{lemma}
\begin{proof}
See Appendix \ref{section_appendix_survey_sampling} for proof.
\end{proof}

According to Eq. (\ref{equation_mean_with_strata}), the estimate of total mean of stratified sampling is weighted by the relative sizes of strata:
\begin{align}\label{equation_mean_estimate_stratified_sampling}
\widehat{\mu} = \sum_{k=1}^K \frac{N_k}{N} \widehat{\mu}_k.
\end{align}

It is noteworthy that the first and second terms in Eq. (\ref{equation_variance_with_strata}) are the within-stratum variance and the between-stratum variance, respectively. This shows that the variance is composed of within and between-stratum variances. 
The concept of within- and between-stratum variances has been widely used in the literature of Fisher discriminant analysis \cite{ghojogh2019fisher}. 

\begin{proposition}\label{proposition_expectation_variance_of_mean_stratified_sampling}
The expectation and variance of the mean of sample by stratified sampling, i.e. Eq. (\ref{equation_mean_estimate_stratified_sampling}), are \cite{barnett1974elements}:
\begin{align}
&\mathbb{E}(\widehat{\mu}) = \mu, \label{equation_stratifiedSampling_mean_expectation_total} \\
&\mathbb{V}\text{ar}(\widehat{\mu}) = \sum_{k=1}^K \big(\frac{N_k}{N}\big)^2 \big(1 - \frac{n_k}{N_k}\big)\, \frac{\sigma_k^2}{n_k}, \label{equation_stratifiedSampling_mean_variance_total}
\end{align}
respectively, where $\mu$ and $\sigma_k^2$ are the mean of whole data and the variance of $k$-th stratum, defined by Eqs. (\ref{equation_mean_estimate}) and (\ref{equation_variance_unbiased_stratum}), respectively.  
\end{proposition}
\begin{proof}
See Appendix \ref{section_appendix_survey_sampling} for proof.
\end{proof}

\begin{corollary}\label{corollary_stratified_sampling_varianceOfMean_proportional}
If the sampling size is proportional to the sizes of strata (called \textbf{proportional allocation} \cite{sukhatme1975allocation}), i.e. 
\begin{align}\label{equation_proportional_allocation}
\frac{n_k}{n} = \frac{N_k}{N},
\end{align}
the variance of mean of sample by stratified sampling is:
\begin{align}\label{equation_stratifiedSampling_varianceOfMean_proportional}
\mathbb{V}\text{ar}(\widehat{\mu}) = \frac{1}{n} (1 - \frac{n}{N}) \sum_{k=1}^K \frac{N_k}{N} \sigma_k^2.
\end{align}
Note that this allocation is especially very useful when classes or strata are imbalanced \cite{he2013imbalanced}.
\end{corollary}
\begin{proof}
See Appendix \ref{section_appendix_survey_sampling} for proof.
\end{proof}

Note that we usually have $N \gg 1$ and $N_k \gg 1$ which results in the following approximation of Eq. (\ref{equation_SRS_varianceOfMean_proportional}):
\begin{align}\label{equation_SRS_varianceOfMean_proportional_approximation}
\mathbb{V}\text{ar}(\widehat{\mu}) \approx \frac{1}{n} (1 - \frac{n}{N}) &\bigg[ \sum_{k=1}^K \frac{N_k}{N} \sigma_k^2 \nonumber \\
&+ \sum_{k=1}^K \frac{N_k}{N-1} (\mu_k - \mu)^2 \bigg],
\end{align}
in SRS. 

\begin{corollary}\label{corollary_stratified_sampling_better_than_SRS}
Stratified sampling always improves the variance of estimation over SRS. This improvement is better if the strata are very different from one another. 
Hence, in stratified sampling, it is better to use strata with different characteristics or variation. 
\end{corollary}
\begin{proof}
Compare Eq. (\ref{equation_SRS_varianceOfMean_proportional_approximation}) in SRS with Eq. (\ref{equation_stratifiedSampling_varianceOfMean_proportional}) in stratified sampling. The variance of estimate of mean by SRS has an additional second term which is non-negative. 
This means that stratified sampling always reduces the variance of estimation and in the worst case, it does not improve over SRS if all the means of strata are equal to the total mean (i.e., if all strata are very similar). 

Moreover, this second term is the between-strata variance, which is also seen in Fisher discriminant analysis \cite{ghojogh2019fisher}. This shows that if the strata are very different (i.e., if the means of strata are very different from each other), the second term gets bold and the improvement of stratified sampling over SRS gets better.  
\end{proof}

It is noteworthy that the \textbf{proportional allocation} \cite{sukhatme1975allocation} is not necessarily an optimal allocation of sampling sizes per stratum. There is an optimal allocation, named \textbf{Neyman allocation} \cite{bankier1988power}, which tries to allocate the sampling sizes for every stratum $k$ in a way that it minimizes the variance of estimation, i.e., Eq. (\ref{equation_stratifiedSampling_mean_variance_total}):
\begin{equation}
\begin{aligned}
& \underset{\{n_1, \dots, n_K\}}{\text{minimize}}
& & \sum_{k=1}^K \big(\frac{N_k}{N}\big)^2 \big(1 - \frac{n_k}{N_k}\big)\, \frac{\sigma_k^2}{n_k}, \\
& \text{subject to}
& & n_1 + \dots + n_K = n,
\end{aligned}
\end{equation}
which is a discrete optimization task in combinatorial optimization \cite{wolsey1999integer}. 

\subsection{Cluster Sampling}

\begin{definition}
We use the same Definition \ref{definition_strata} for defining clusters. In cluster sampling, the strata are referred to as clusters. 
\end{definition}

\begin{definition}[Cluster Sampling]\label{definition_cluster_sampling}
Let data $\mathcal{D}$ consist of $K$ clusters defined in Eq. (\ref{equation_disjoint_sets}). 
Cluster sampling is simple random sampling (see Definition \ref{definition_SRS}) of size $c < K$ clusters, where all instances of the selected clusters are taken in the sample \cite{barnett1974elements}. 
An illustration of cluster sampling is shown in Fig. \ref{figure_survey_sampling}.
\end{definition}

Therefore, the sample is composed of the sampled clusters. Suppose $K=5$ and the clusters $\mathcal{D}_1$, $\mathcal{D}_3$, and $\mathcal{D}_4$ are sampled; then, the sample would be:
\begin{align}
\mathcal{S} = \mathcal{D}_1 \oplus \mathcal{D}_3 \oplus \mathcal{D}_4.
\end{align}

\begin{example}
The following example clarifies the difference of SRS, bootstrapping, stratified sampling, and cluster sampling. We want to do a survey in the city, asking people some questions. In SRS, we randomly find people in the city and ask them questions. In bootstrapping, we do not record the names of already asked people; thus, there is a possibility that some people are asked more than once. We consider houses of city as strata or clusters. In stratified sampling, we go to every house and randomly interview with some people in each house. In cluster sampling, however, we sample some houses -- rather than going to all houses -- and interview with all people in the selected houses -- rather than sampling people in the houses. This example shows that cluster sampling is for convenience because sampling from houses is much easier than sampling from people in the houses. 
\end{example}

\begin{corollary}\label{corollary_mean_variance_with_clusters}
Suppose data consist of $K$ strata. According to Eqs. (\ref{equation_mean_estimate}) and (\ref{equation_variance_unbiased}), the actual mean and variance of the $k$-th stratum are as Eqs. (\ref{equation_mean_unbiased_stratum}) and (\ref{equation_variance_unbiased_stratum}), respectively. 
\end{corollary}

Again, the mean of data with $K$ clusters is as in Eq. (\ref{equation_mean_with_strata}). 
According to the HT estimator, the estimate of mean by cluster sampling is:
\begin{align}
\widehat{\mu} &\overset{(\ref{equation_mean_with_strata})}{=} \sum_{\mathcal{D}_k \in \mathcal{S}} \frac{(N_k / N) \mu_k}{c/K} = \frac{K}{N} \Big( \frac{1}{c} \sum_{\mathcal{D}_k \in \mathcal{S}} N_k\, \mu_k \Big) \nonumber \\
&= \frac{K}{N} \Big( \frac{1}{c} \sum_{\mathcal{D}_k \in \mathcal{S}} \tau_k \Big), \label{equation_mean_estimate_cluster_sampling}
\end{align}
where:
\begin{align}
\tau_k := N_k\, \mu_k. 
\end{align}
According to Proposition \ref{proposition_HT_unbiased}, this is an unbiased estimator of mean.
The term within the parentheses in Eq. (\ref{equation_mean_estimate_cluster_sampling}) is SRS in the cluster level, which makes sense because we have SRS in the cluster level according to the definition of cluster sampling. Hence, if $\widehat{\mu}_* := (1/c) \sum_{\mathcal{D}_k \in \mathcal{S}} \tau_k$ denotes the estimate of mean in the cluster level, Eq. (\ref{equation_mean_estimate_cluster_sampling}) becomes:
\begin{align}
\widehat{\mu} = \frac{K}{N}\, \widehat{\mu}_*. \label{equation_mean_estimate_cluster_sampling_2}
\end{align}

\begin{proposition}\label{proposition_expectation_variance_of_mean_cluster_sampling}
The expectation and variance of the mean of sample by cluster sampling, i.e. Eq. (\ref{equation_mean_estimate_cluster_sampling}), are \cite{barnett1974elements}:
\begin{align}
&\mathbb{E}(\widehat{\mu}) = \mu, \label{equation_clusterSampling_mean_expectation_total} \\
&\mathbb{V}\text{ar}(\widehat{\mu}) = \frac{K^2}{N^2} (1 - \frac{c}{K}) \frac{\sigma_*^2}{c}, \label{equation_clusterSampling_mean_variance_total} 
\end{align}
respectively, where:
\begin{align}
\sigma_*^2 = \frac{1}{K-1} \sum_{k=1}^K (\tau_k - \frac{1}{K} \sum_{k'=1}^K \tau_{k'})^2,
\end{align}
and $\mu$ and $\sigma_k^2$ are the mean of whole data and the variance of $k$-th cluster, defined by Eqs. (\ref{equation_mean_estimate}) and (\ref{equation_variance_unbiased_stratum}), respectively.  
\end{proposition}
\begin{proof}
See Appendix \ref{section_appendix_survey_sampling} for proof.
\end{proof}

\begin{corollary}\label{corollary_cluster_sampling_varianceOfMean_equalClusters}
If the size of clusters are equal, i.e. $N_k = L, \forall k \in \{1, \dots, K\}$, which results in:
\begin{align}\label{equation_euqal_cluster_size}
N = KL \implies \frac{K^2}{N^2} = \frac{1}{L^2},
\end{align}
the variance of mean of sample by cluster sampling is:
\begin{align}\label{equation_clusterSampling_varianceOfMean_equalClsuters}
\mathbb{V}\text{ar}(\widehat{\mu}) = \frac{1}{c} \big(1 - \frac{c}{K}\big) \bigg[ \frac{1}{K-1} \sum_{k=1}^K (\mu_k - \mu)^2 \bigg].
\end{align}
\end{corollary}
\begin{proof}
See Appendix \ref{section_appendix_survey_sampling} for proof.
\end{proof}

For comparison of cluster sampling with SRS, we consider the same sample size $n = c\, L$ in SRS. 
According to Eq. (\ref{equation_SRS_mean_variance}), the variance of estimate of mean by SRS, with the sample size $n = c\, L$, is:
\begin{align*}
\mathbb{V}\text{ar}(\widehat{\mu}) &= \big(1 - \frac{n}{N}\big)\, \frac{\sigma^2}{n} = \big(1 - \frac{c\, L}{K L}\big)\, \frac{\sigma^2}{c\, L} \\
&\overset{(\ref{equation_variance_with_strata})}{=} \frac{1}{c} (1 - \frac{c}{K}) \frac{1}{L} \Big[ \sum_{k=1}^K \frac{L-1}{KL-1} \sigma_k^2 \\
&~~~~~~~~~~~~~~ + \sum_{k=1}^K \frac{L}{KL-1} (\mu_k - \mu)^2 \Big].
\end{align*}
As we usually we have $L \gg 1$ and $M \gg 1$, this equation can be approximated as:
\begin{equation}\label{equation_clusterSampling_varianceOfMean_equalClsuters_comparedToSRS}
\begin{aligned}
&\mathbb{V}\text{ar}(\widehat{\mu}) =\\
&\frac{1}{c} (1 - \frac{c}{K}) \frac{1}{L} \Big[ \sum_{k=1}^K \frac{1}{K} \sigma_k^2 + \sum_{k=1}^K \frac{1}{K-1} (\mu_k - \mu)^2 \Big].
\end{aligned}
\end{equation}

\begin{corollary}\label{corollary_cluster_sampling_better_than_SRS}
Comparing Eqs. (\ref{equation_clusterSampling_varianceOfMean_equalClsuters}) and (\ref{equation_clusterSampling_varianceOfMean_equalClsuters_comparedToSRS}) shows that the cluster sampling can be better than SRS and this improvement can be better if the clusters are more similar in terms of their means, $\mu_k$. Note that, in contrast to stratified sampling, cluster sampling is not necessarily better than SRS because of division by $L$. 
\end{corollary}

\begin{example}
The following example, whose credit is for \cite{zhu2017lectureSurveySampling}, shows when stratified sampling and when cluster sampling are better to use. 
Assume we have a dataset $\mathcal{D} = \{1,2,3,1,2,3,1,2,3\}$. On one hand, a good set of strata is $\mathcal{D}_1 = \{1,1,1\}$, $\mathcal{D}_2 = \{2,2,2\}$, and $\mathcal{D}_3 = \{3,3,3\}$ because, according to Corollary \ref{corollary_stratified_sampling_better_than_SRS}, the strata are very different in terms of their means. On the other hand, a good set of clusters is $\mathcal{D}_1 = \{1,2,3\}$, $\mathcal{D}_2 = \{1,2,3\}$, and $\mathcal{D}_3 = \{1,2,3\}$ because, according to Corollary \ref{corollary_cluster_sampling_better_than_SRS}, the clusters are very similar in terms of their means.
These make sense because stratified sampling samples by SRS from every stratum while cluster sampling samples by SRS from the clusters and takes all samples of the selected clusters. 
\end{example}

\subsection{More Advanced Survey Sampling}

\subsubsection{Multistage Sampling}

\begin{definition}[Multistage sampling]
As its name clarifies, multistage sampling \cite{lance2016sampling} draws samples stage-wise where at each stage or level, the population to sample from gets smaller.   
\end{definition}

Multistage sampling divides data into clusters or strata stage-wise and samples within them. 
An example of multistage sampling is cluster sampling in the first stage and then performing SRS within every sampled cluster. Using multistage sampling, we can combine many different survey sampling methods. 

\subsubsection{Network Sampling}

\begin{definition}[Network sampling]
Network sampling \cite{granovetter1976network,frank1977survey,frank2011survey} refers to sampling from a family of networks. Consider a graph $G = (V, E)$ where $V$ is the set of vertices and $E \subseteq V \times V$ is the set of edges. We can have different sub-networks of $G$. Let $F$ denotes the set of sub-networks of $G$. We refer to $G$ as the \textit{population graph} or the \textit{population network}. Sampling networks from the set $F$ is named \textit{network sampling} \cite{schofield2011network}.
\end{definition}

Note that network sampling is a family of methods and not merely one sampling algorithm \cite{heckathorn2017network}. 
There are also some network sampling methods for streaming networks \cite{ahmed2013network}.

\subsubsection{Snowball Sampling}

\begin{definition}[Snowball sampling]
Snowball sampling \cite{goodman1961snowball} has two steps. First, it identifies several potential samples or candidates. Then, the selected samples/candidates select some other samples/candidates based on their own judgments. Its name comes from the analogy of a snowball which gets bigger and bigger by rolling down a hill; here, the sample size also gets larger and larger exponentially. 
\end{definition}

Snowball sampling can be considered as a spacial case of network sampling. 
It can be used in social analysis and sociology \cite{heckathorn2017network}, where the judgement of selected people draws samples in the survey. For example, a private survey is conducted in the social media where people are invited to it. In programming and mathematics, one can write rules (e.g., see fuzzy logic \cite{klir1995fuzzy}) for selecting new samples by the already selected samples. In these cases, as there is no probability involved, snowball sampling is a non-probability sampling method \cite{vehovar2016non}. However, one may want to write the rules of selecting samples stochastically using probability.


\section{Sampling from Distribution: Monte Carlo Methods}\label{section_Monte_Carlo}

Sampling can also be done by sampling from a probability distribution. If the distribution is a simple distribution or if we can have the Cumulative Distribution Function (CDF), we can easily sample from distribution. However, if the distribution is complicated, we cannot simply and directly sample from them. In these situations, we use the Monte Carlo (MC) methods \cite{hammersley2013monte,kalos2009monte}. MC methods can be divided into simple MC methods, Markov Chain Monte Carlo (MCMC) methods, and efficient MC methods \cite{mackay2003information}. In the following, we explain these different methods in detail.  

\subsection{Sampling from Inverse Cumulative Distribution Function}

In some cases, we can easily have the inverse CDF of distribution. An example is dealing with one dimensional distributions where we can easily plot the inverse CDF. It is noteworthy that the inverse CDF is also refered to as the quantile function \cite{parzen1979nonparametric}. One can sample from a distribution using the inverse CDF or the quantile function. Assume the distribution is one dimensional. A random number is drawn from the uniform distribution $U(0,1)$. Feeding this random number to the inverse CDF gives us a random number drawn from the distribution. 
This procedure is illustrated in Fig. \ref{figure_CDF_sampling}. 
It is noteworthy that this type of sampling makes sense because, as Fig. \ref{figure_CDF_sampling} shows, it draws more samples from the modes of distribution as expected. 
This is basic sampling approach used in many statistical methods (e.g., see \cite{shaw2006sampling}). 

\begin{figure}[!t]
\centering
\includegraphics[width=3.2in]{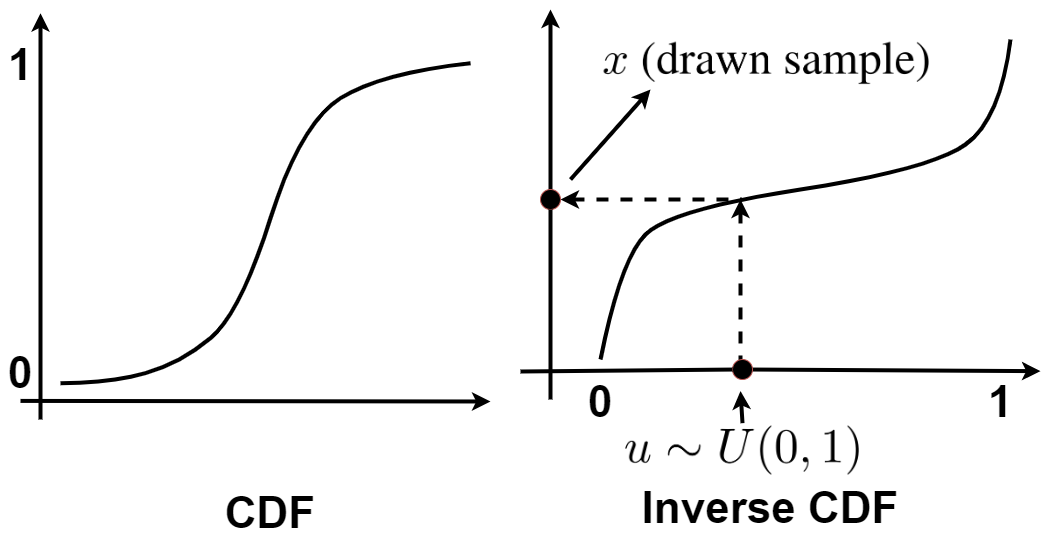}
\caption{Sampling from inverse CDF}
\label{figure_CDF_sampling}
\end{figure}

\subsection{Monte Carlo Approximation}

\subsubsection{Definition}

Suppose we are considering some $d$-dimensional data $x \in \mathbb{R}^d$. Let $f(x)$ be the Probability Density Function (PDF) of data. 
Consider $h(x)$ is a function over the data $x$. 
According to definition, the expectation of function $h(x)$ over the distribution $f(x)$ and the probability of function $h(x)$ belonging to a set $\mathcal{A}$ are:
\begin{align}
&\mathbb{E}(h(x)) = \int h(x)\, f(x)\, dx, \label{equation_MC_approximate_expectation_exact}\\
&\mathbb{P}(h(z) \in \mathcal{A}) = \int_{h(x) \in \mathcal{A}} f(x)\, dx, \label{equation_MC_approximate_prob_exact}
\end{align}
respectively. 

\begin{definition}[Monte Carlo approximation]
Using a sample of size $n$ from distribution $f(x)$ (i.e., $\{x_1, \dots, x_n\} \sim f(x)$), we can approximate Eqs. (\ref{equation_MC_approximate_expectation_exact}) and (\ref{equation_MC_approximate_prob_exact}) by:
\begin{align}
&\mathbb{E}(h(x)) \approx \frac{1}{n} \sum_{i=1}^n h(x_i), \label{equation_MC_approximate_expectation_approximate}\\
&\mathbb{P}(h(z) \in \mathcal{A}) \approx \frac{1}{n} \sum_{i=1}^n \mathbb{I}\big(h(x_i) \in \mathcal{A}\big), \label{equation_MC_approximate_prob_approximate}
\end{align}
where $\mathbb{I}(\cdot)$ denotes the indicator function which is one and zero when its condition is and is not satisfied, respectively. 
\end{definition}

As the above definition states, the MC approximation generates many samples from the distribution in order to approximate the expectation by mean (or average) of the samples.
Obviously, the more the $n$ is, the better the approximation becomes.

\begin{example}
We can approximate the $\pi$ number using the Monte Carlo approximation \cite{kalos2009monte}. As Fig. \ref{figure_MC_approximation_pi} shows, consider a square with length one. A quarter of circle exists within the square with radius one. If we generate many samples uniformly from inside of the square, we see that the proportion of samples which fall within the quarter of circle (green circle points) to the entire samples (both green circles and red squares) approximately goes to $\pi / 4$ as expected. The more samples we generate, the closer this proportion gets to $\pi / 4$. 
\end{example}

\begin{figure}[!t]
\centering
\includegraphics[width=1.5in]{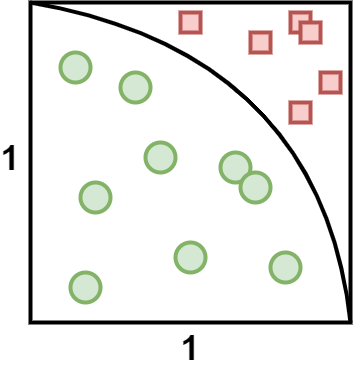}
\caption{Approximating $\pi$ with Monte Carlo approximation}
\label{figure_MC_approximation_pi}
\end{figure}

\subsubsection{Where the Name Came From?}

It is noteworthy to briefly mention where the name ``Monte Carlo'' came from. Some Monte Carlo Markov Chain method, with an approximation approach, was proposed by a physicist named Stanislaw Ulam. Then, John von Neumann also joined him in his work. The work of these two required a code name. One of their colleagues, named Nicholas Metropolis, suggested the name ``Monte Carlo'' referring to the Monte Carlo Casino in Monaco where Ulam's uncle used to borrow money from relatives to gamble; note that gambling is related to probabilistic approach of this method. Hence, they named this technique Monte Carlo \cite{mazhdrakov2018monte}.

\subsection{Simple Monte Carlo Methods}

The MC methods are iterative methods which generate samples from a distribution.
Some of the MC methods are the simple MC methods. These methods draw samples blindly like a  blindfold person because every step or iteration does not depend on the previous iteration. Therefore, the iterations are independent and are performed blindly in the space of data/distribution \cite{mackay2003information}. Some important methods in this category are importance sampling and rejection sampling, explained in the following. 

\subsubsection{Importance Sampling}

Consider a distribution which may be complicated. We can write the Probability Density Function (PDF) or Probability Mass Function (PMF) of this distribution as:
\begin{align}\label{equation_PDF_nasty}
f(X) = \frac{P^*(X)}{Z},
\end{align}
where $Z$ is the marginal distribution or the normalizing factor which can be intractable to compute because of integrating/summing over all domain of data. 
Note that the normalizing factor $Z$ is also called the \textit{partition function} in physical models such as the Ising model \cite{cipra1987introduction,mccoy2014two}. 
The numerator, $F^*(X)$, is the scaled or non-normalized PDF/PMF of distribution and does not necessarily integrate/sum to one but has the shape of distribution. 

\begin{definition}[Importance sampling]\label{definition_importance_sampling}
Consider a function of interest, denoted by $h(X)$. We want to calculate the expectation of this function $h(X)$ on data, over the distribution $f(X)$ or $P^*(X)$. However, as the distribution is complicated and hard to compute, we can estimate this expectation using another simple distribution $Q(X)$. This simple distribution, which we can easily draw samples from, can be any distribution such as uniform or Gaussian. Importance sampling \cite{glynn1989importance} performs this estimation. 
\end{definition}

\begin{proposition}\label{proposition_importance_sampling_expectation}
In importance sampling, we sample from the simple distribution $Q(X)$ rather than sampling from the complicated distribution $P^*(X)$ which is very hard to do. 
First, consider the average of function $h(X)$ on the $n$ samples $\{x_i\}_{i=1}^n$ drawn from $Q(X)$, which is $(1/n) \sum_{i=1}^n h(x_i)$. However, this expression is not yet the desired expectation (see Definition \ref{definition_importance_sampling}) because the samples are drawn from $Q(X)$ rather than $P^*(X)$. To make it an estimation of the desired expectation, we should weight the instances in this summation as:
\begin{align}\label{equation_importance_sampling_expectation}
\mathbb{E}_{x\sim f(X)}(h(x)) \approx \frac{1}{\sum_{j=1}^{n} \frac{P^*(x_j)}{Q(x_j)}} \sum_{i=1}^{n} \frac{P^*(x_i)}{Q(x_i)} h(x_i),
\end{align}
which gets more accurate by increasing the sample size $n$.
\end{proposition}
\begin{proof}
See Appendix \ref{section_appendix_Monte_Carlo} for proof. 
\end{proof}

It is noteworthy that importance sampling in statistics is related to the umbrella sampling \cite{kumar1992weighted} in physics. 
Moreover, a recent improvement over the importance sampling is the Annealed Importance Sampling. We refer the readers to paper \cite{neal2001annealed} for more information about it. 

\subsubsection{Rejection Sampling}


Assume we want to draw samples from a complicated distribution $f(X)$ or its non-normalized version $P^*(X)$. Rejection sampling
\cite{casella2004generalized,bishop2006pattern,robert2013monte} can be used to draw samples from a simple distribution $Q(X)$, instead, and use those samples to generate samples drawn from $P^*(X)$. 

\begin{definition}[Rejection sampling]
In rejection sampling \cite{casella2004generalized}, we consider a simple-to-sample distribution denoted by $Q(X)$ where, for a positive number $c$, we have:
\begin{align}\label{equation_rejection_sampling_condition}
c\, Q(x) \geq P^*(x), \quad \forall x \in \textbf{dom}(X),
\end{align}
where $\textbf{dom}(X)$ denotes the domain of distribution or the range of data $X$. For sampling $x_i$ from the complicated distribution $P^*(X)$ (see Eq. (\ref{equation_PDF_nasty})), we draw sample from the simple distribution $Q(X)$, i.e., $x_i \sim Q(X)$. Then, we sample a number $u_i$ from the uniform distribution $U(0, c\, Q(x_i))$. If this $u_i$ is smaller than $P^*(x_i)$, it is accepted to be the sample from $P^*(X)$; otherwise, we reject it and repeat this procedure. The algorithm and illustration of rejection sampling can be seen in Algorithm \ref{algorithm_rejection_sampling} and Fig. \ref{figure_Rejection_sampling}, respectively. 
\end{definition}

\SetAlCapSkip{0.5em}
\IncMargin{0.8em}
\begin{algorithm2e}[!t]
\DontPrintSemicolon
    \textbf{Input:} $P^*(X), Q(X), c$\;
    \textbf{Output:} $\mathcal{S} = \{x_i\}_{i=1}^n \sim P^*(X)$\;
    $\mathcal{S} \gets \varnothing$\;
    \For{sample index $i$ from $1$ to $n$}{
        $x_i \sim Q(X)$\;
        $u_i \sim U(0, c\, Q(x_i))$\;
        \uIf{$u_i < P^*(x_i)$}{
            Accept $x_i$: $\mathcal{S} \gets \mathcal{S} \cup \{x_i\}$\;
        }
        \Else{
            Reject $x_i$: $i \gets i - 1$\;
        }
    }
\caption{Rejection sampling}\label{algorithm_rejection_sampling}
\end{algorithm2e}
\DecMargin{0.8em}

\begin{figure}[!t]
\centering
\includegraphics[width=2.8in]{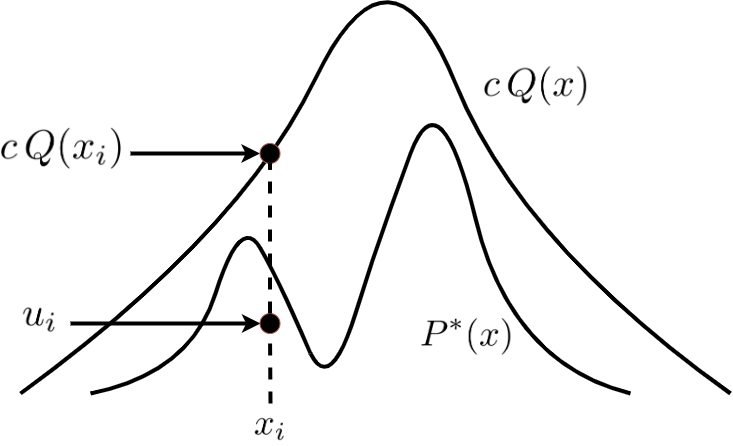}
\caption{Rejection sampling}
\label{figure_Rejection_sampling}
\end{figure}

One challenge in rejection sampling is finding the appropriate value for $c$. The larger $c$ helps in satisfying Eq. (\ref{equation_rejection_sampling_condition}) but results in many number of rejections because of the condition $u_i < P^*(x_i)$. Therefore, there is trade-off here. One may use Algorithm \ref{algorithm_rejection_sampling_finding_c} to find an appropriate $c$ alongside sampling; although this algorithm has many redundancy because a non-valid $c$ requires re-sampling from the scratch. Moreover, the sample size $n$ should be large in Algorithm \ref{algorithm_rejection_sampling_finding_c} to check Eq. (\ref{equation_rejection_sampling_condition}) for most points in $\textbf{dom}(X)$. Note that this algorithm requires a valid $Q(X)$ which satisfies Eq. (\ref{equation_rejection_sampling_condition}) for some $c$ eventually. 

\SetAlCapSkip{0.5em}
\IncMargin{0.8em}
\begin{algorithm2e}[!t]
\DontPrintSemicolon
    \textbf{Input:} $P^*(X), Q(X), c$\;
    \textbf{Output:} $\mathcal{S} = \{x_i\}_{i=1}^n \sim P^*(X)$\;
    $c \gets $ initial small $c$\;
    cIsValid $\gets$ False\;
    \While{\textbf{not} cIsValid}{
        $\mathcal{S} \gets \varnothing$\;
        cIsValid $\gets$ True\;
        \For{sample index $i$ from $1$ to $n$}{
            $x_i \sim Q(X)$\;
            \If{$c\, Q(x_i) < P^*(x_i)$}{
                cIsValid $\gets$ False\;
                Increase $c$ a little\;
                Break the for loop\;
            }
            $u_i \sim U(0, c\, Q(x_i))$\;
            \uIf{$u_i < P^*(x_i)$}{
                Accept $x_i$: $\mathcal{S} \gets \mathcal{S} \cup \{x_i\}$\;
            }
            \Else{
                Reject $x_i$: $i \gets i - 1$\;
            }
        }    
    }
\caption{Rejection sampling with simultaneous calculation of $c$}\label{algorithm_rejection_sampling_finding_c}
\end{algorithm2e}
\DecMargin{0.8em}

There exist more advanced versions of rejection sampling recently proposed in the literature. Some of these methods are adaptive rejection sampling \cite{gilks1992adaptive,gorur2011concave,martino2011generalization}, ensemble rejection sampling \cite{deligiannidis2020ensemble}, discriminator rejection sampling \cite{azadi2018discriminator}, and variational rejection sampling \cite{grover2018variational}, which we do not cover in this paper and refer the readers to them for more information. 

\subsection{Markov Chain Monte Carlo Methods}

The second category of Monte Carlo methods is Markov Chain Monte Carlo (MCMC) methods \cite{mackay2003information,brooks2011handbook,geyer2011introduction}. In MCMC methods, in contrast to the simple Monte Carlo methods, iterations are not independent and blindly sampled but every iteration/step of Monte Carlo is dependent to its previous iteration/step. This feature is referred to as the Markov property, already explained in Section \ref{section_Markov_property}. 

\subsubsection{Metropolis Algorithm}

\begin{definition}[Metropolis algorithm]
Using the Metropolis algorithm, proposed by Metropolis et. al. \cite{metropolis1953equation}, we can sample from a complicated distribution, denoted by $f(X)$ or $P^*(X)$ (see Eq. (\ref{equation_PDF_nasty})), using a simple distribution $Q$ as the proposal function. As Algorithm \ref{algorithm_metopolis} shows, we start from a random number/vector in the range of data. Then, we draw the next sample, based on the current location, using a simple conditional distribution $Q(X_{i+1}; X_i)$ as the proposal function. This proposal function is symmetric, i.e.:
\begin{align}\label{equation_Metropolis_symmetric_proposal_function}
Q(x_{i+1}; x_i) = Q(x_i; x_{i+1}).
\end{align}
With the probability:
\begin{align}\label{equation_Metropolis_p_accept}
p_\text{accept} = \min\Big(\frac{P^*(x_i)}{P^*(x_{i-1})}, 1\Big),
\end{align}
we accept the proposed sample $x_i$; otherwise we reject it. This procedure is repeated until we have all the $n$ samples. 
The procedure of Metropolis algorithm is depicted in Fig. \ref{figure_Metropolis_algorithm}. As this figure shows, more samples are drawn from modes of $P^*(X)$, as expected. 
\end{definition}

\begin{figure}[!t]
\centering
\includegraphics[width=2.8in]{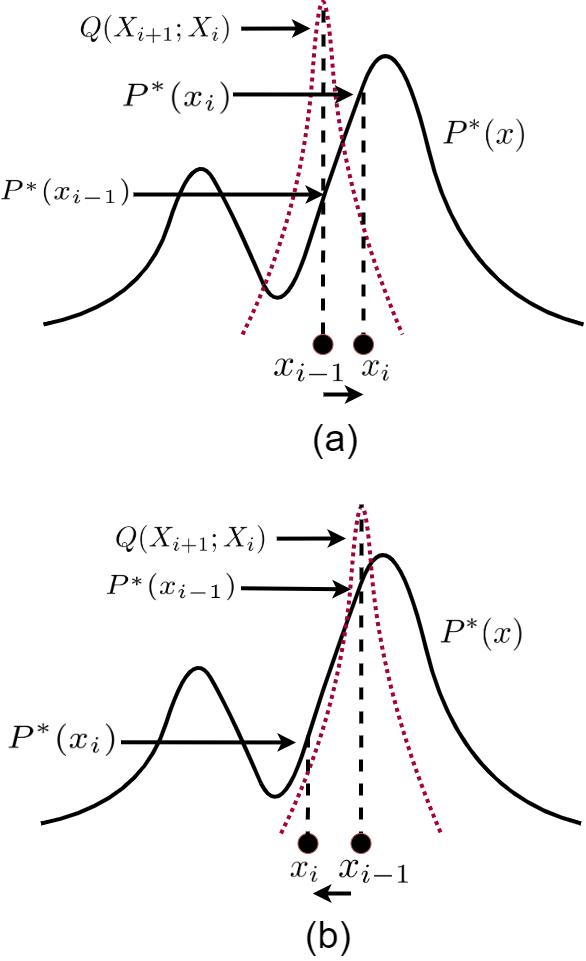}
\caption{Metropolis algorithm for MCMC sampling: (a) $x_i$ is accepted as a drawn sample from $P^*(X)$, (b) The previously drawn sample becomes $x_{i-1}$ in this iteration. The $x_i$ may be accepted as a drawn sample from $P^*(X)$ with probability $P^*(x_i) / P^*(x_{i-1})$.}
\label{figure_Metropolis_algorithm}
\end{figure}

\begin{remark}
Usually, the normalization factor or the partition function $Z$ is computationally expensive to calculate because of the integral or summation over all values. In the distributions where $Z$ does not depend on the point $x_i$, the Metropolis algorithm has the advantage of not requiring to compute $Z$ because in Eq. (\ref{equation_Metropolis_p_accept}), the normalization factors $Z$ are cancelled from the terms in the numerator and denominator.
\end{remark}

An example of the proposal function $Q$ is a Gaussian distribution:
\begin{align}\label{equation_MCMC_proposal_Gaussian}
\mathbb{R}^{d} \ni x_i := x_{i-1} + \mathcal{N}(\b{0}, \sigma^2 \b{I}),
\end{align}
where $\b{I} \in \mathbb{R}^{d \times d}$ is the identity matrix and $\sigma$ determines the step size. The following remark discusses the effect of values for $\sigma$. 

\begin{remark}
In Eq. (\ref{equation_MCMC_proposal_Gaussian}), the appropriate value for $\sigma$ or the step size can be challenging to find. It has been shown in \cite{rosenthal2014optimising} that a good value for $\sigma$ for most cases is $\sigma = 2.38$. Less than this value, e.g. $\sigma=0.1$, results in very small step sizes and a very slow progress of algorithm. Larger than this value, e.g. $\sigma = 25$, results in very large step sizes and many rejections in the Metropolis algorithm because we may jump to very low-probability values in $P^*(X)$ with large step sizes. Another related paper in finding the best $\sigma$ value is \cite{roberts2001optimal}. 
\end{remark}

\SetAlCapSkip{0.5em}
\IncMargin{0.8em}
\begin{algorithm2e}[!t]
\DontPrintSemicolon
    \textbf{Input:} $P^*(X), Q(X_{i+1}; X_{i}), c$\;
    \textbf{Output:} $\mathcal{S} = \{x_i\}_{i=1}^n \sim P^*(X)$\;
    $\mathcal{S} \gets \varnothing$\;
    $x_0 \gets $ a random number/vector in $\textbf{dom}(X)$\;
    \For{sample index $i$ from $1$ to $n$}{
        $x_i \sim Q(X; x_{i-1})$\;
        $p_\text{accept} = \min(\frac{P^*(x_i)}{P^*(x_{i-1})}, 1)$\;
        $u_i \sim U(0,1)$\;
        \uIf{$u_i < p_\text{accept}$}{
            Accept $x_i$: $\mathcal{S} \gets \mathcal{S} \cup \{x_i\}$\;
        }
        \Else{
            Reject $x_i$: $i \gets i - 1$\;
        }
    }
\caption{Metropolis algorithm for MCMC sampling}\label{algorithm_metopolis}
\end{algorithm2e}
\DecMargin{0.8em}

\begin{definition}[Stationary distribution]
Consider a Markov chain (see Section \ref{section_Markov_property}) with the transition function $A(v; u)$ as the probability of transition from state $u$ to state $v$. The Markov chain has a stationary distribution if integrating/summing over all transitions from other states to state $v$ is equal to the probability of state $v$:
\begin{align}\label{equation_stationary_distribution_sumInward}
\int \mathbb{P}(u)\, A(v; u)\, du = \mathbb{P}(v).
\end{align}
In other words, a stationary distribution satisfies \cite{ross1996stochastic,parzen1999stochastic}:
\begin{align}
\mathbb{P}_{t-1}(u) = \mathbb{P}_t(u), \quad \forall u, t.
\end{align}
\end{definition}

\begin{lemma}\label{lemma_stationary_balance_condition}
Consider a Markov chain (see Section \ref{section_Markov_property}) with the transition function $Q(v; u)$. If the probability of states in the Markov chain satisfy the ``balance condition'':
\begin{align}\label{equation_stationary_balance_condition}
\mathbb{P}(u)\, A(v; u) = \mathbb{P}(v)\, A(u; v), \quad \forall u,v,
\end{align}
then, it is a stationary distribution of this Markov chain.
\end{lemma}
\begin{proof}
See Appendix \ref{section_appendix_Monte_Carlo} for proof. 
\end{proof}

The Metropolis algorithm can be seen as a Markov chain and that is why it is in the category of MCMC methods. 
The transition function in Metropolis algorithm is the multiplication of the proposal function $Q(X_{i}; X_{i-1})$ and the probability of acceptance of proposal (i.e., Eq. (\ref{equation_Metropolis_p_accept})):
\begin{align}\label{equation_Metropolis_transition_function}
A(x_{i}; x_{i-1}) = Q(x_{i}; x_{i-1}) \times \min\Big(\frac{P^*(x_i)}{P^*(x_{i-1})}, 1\Big).
\end{align}

\begin{proposition}\label{proposition_Metropolis_symmetric}
The transition function in the Metropolis algorithm, which is Eq. (\ref{equation_Metropolis_transition_function}), is symmetric with respect to the previous sample $x_{i-1}$ and the new sample $x_i$. And therefore, the Metropolis algorithm has a stationary distribution. 
\end{proposition}
\begin{proof}
See Appendix \ref{section_appendix_Monte_Carlo} for proof. 
\end{proof}

\subsubsection{Metropolis-Hastings Algorithm}

Hastings generalized the Metropolis algorithm to not necessarily symmetric proposal function. We call this algorithm the Metropolis-Hastings algorithm \cite{hastings1970monte}. Some papers refer to this method as the Metropolis algorithm, however.  
The difference of the Metropolis-Hastings algorithm from the Metropolis algorithm is in the probability of acceptance of the proposed sample. In other words, Eq. (\ref{equation_Metropolis_p_accept}) is replaced by:
\begin{align}\label{equation_Metropolis_Hastings_p_accept}
p_\text{accept} = \min\Big(\frac{P^*(x_i)\, Q(x_{i-1}; x_i)}{P^*(x_{i-1})\, Q(x_i; x_{i-1})}, 1\Big).
\end{align}
This equation can be seen as multiplication of two terms:
\begin{align}\label{equation_Metropolis_Hastings_p_accept_2}
p_\text{accept} = \min\Big(\frac{P^*(x_i)}{Q(x_i; x_{i-1})} \times \frac{Q(x_{i-1}; x_i)}{P^*(x_{i-1})}, 1\Big).
\end{align}
Comparing Eq. (\ref{equation_Metropolis_Hastings_p_accept_2}) with Eq. (\ref{equation_importance_sampling_expectation}) shows that it contains the weights used in importance sampling for both the old and newly proposed samples. 

\subsubsection{Gibbs Sampling}

Gibbs sampling, firstly proposed by \cite{geman1984stochastic}, draws samples from a $d$-dimensional multivariate distribution $P^*(X)$ using $d$ conditional distributions  \cite{gelfand2000gibbs,mackay2003information,bishop2006pattern}. This method is named after the physicist Josiah Willard Gibbs. 
The intuition of this method is similar to the coordinate descent in optimization \cite{wu2008coordinate,wright2015coordinate}. It assumes that the conditional distributions of every coordinate/feature/dimension of data conditioned on the rest of coordinates are simple to draw samples from. 

\begin{definition}
In Gibbs sampling, we desire to sample from a multivariate distribution $P^*(X)$ where $X \in \mathbb{R}^d$. We denote:
\begin{align}
\mathbb{R}^d \ni x_i := [x_i^{(1)}, x_i^{(2)}, \dots, x_i^{(d)}]^\top.
\end{align}
We start from a random $d$-dimensional vector int he range of data. Then, we sample the first dimension of the first sample from the distribution of the first dimension conditioned on the other dimensions. We do it for all dimensions, where the $j$-th dimension is samples as:
\begin{align}
x_i^{(j)} \sim P^*(X^{(j)}\, |\, X^{(1)}, \dots, X^{(j-1)}, X^{(j+1)}, \dots, X^{(d)}).
\end{align}
We do this for all dimensions until all dimensions of the first sample are drawn. Then, starting from the first sample, we repeat this procedure for the dimensions of the second sample. We iteratively perform this for all samples; however, some initial samples are not yet valid because the algorithm has started from a not-necessarily valid vector. We accept all samples after some burn-in iterations, denoted by $t_\text{burnIn}$. Algorithm \ref{algorithm_Gibbs_sampling} shows the procedure of Gibbs sampling. 
\end{definition}

One of the challenges of Gibbs sampling is that we do not know exactly what burn-in iteration is appropriate. Large burn-in iteration results in more useless computations and small burn-in iteration may provide us with some not valid samples. However, it has been shown in the literature that Gibbs sampling, as well as Metropolis algorithms, are very fast and usually even a small burn-in iteration works well \cite{dwivedi2018log}.

\SetAlCapSkip{0.5em}
\IncMargin{0.8em}
\begin{algorithm2e}[!t]
\DontPrintSemicolon
    \textbf{Input:} $P^*(X), Q(X_{i+1}; X_{i}), c$\;
    \textbf{Output:} $\mathcal{S} = \{x_i\}_{i=1}^n \sim P^*(X)$\;
    $\mathcal{S} \gets \varnothing$\;
    $x_0 \gets $ random $d$-dimensional vector in $\textbf{dom}(X)$\;
    \For{sample index $i$ from $1$ to $n+t_\text{burnIn}$}{
        $x_i \gets x_{i-1}$\;
        \For{dimension $j$ from $1$ to $d$}{
            $x_i^{(j)} \sim P^*(X^{(j)}\, |\, X^{(1)}, \dots, X^{(j-1)},$ \\
            $\quad \quad \quad \quad \quad \quad \quad \quad X^{(j+1)}, \dots, X^{(d)})$\;
            $x_i \gets [x_i^{(1)}, x_i^{(2)}, \dots, x_i^{(d)} ]^\top$\;
        }
        \If{$i \geq t_\text{burnIn}$}{
            $\mathcal{S} \gets \mathcal{S} \cup \{x_i\}$\;
        }
    }
\caption{Gibbs sampling}\label{algorithm_Gibbs_sampling}
\end{algorithm2e}
\DecMargin{0.8em}

\begin{proposition}\label{proposition_Gibbs_special_case_Metropolis}
Gibbs sampling can be seen as a special case of the Metropolis-Hastings algorithm which accepts the proposed samples with probability one:
\begin{align}
p_\text{accept} = 1.
\end{align}
\end{proposition}
\begin{proof}
See Appendix \ref{section_appendix_Monte_Carlo} for proof. 
\end{proof}

\subsubsection{Slice Sampling}

One of the problems with the Metropolis and Metropolis-Hastings algorithms is not knowing the best step size in the proposal function (e.g., see Eq. (\ref{equation_MCMC_proposal_Gaussian})). Slice sampling, proposed by Neal \cite{neal2003slice} and Skilling \cite{skilling2003slice}, handles this issue by being robust to step size.

Slice sampling is used to draw samples from a complicated distribution $P^*(X)$. 
The algorithm of slice sampling is depicted in Fig. \ref{figure_Slice_sampling}. Initially, a random point is considered in $\textbf{dom}(X)$.
Then, a arbitrary direction (line) in the space of $d$-dimensional data is considered to deal with a one-dimensional distribution (as seen in Fig. \ref{figure_Slice_sampling}). Note that in Gibbs sampling, only the direction along one of the dimensions was considered in the conditional distributions; however, slice sampling gives freedom of choice to user to take any direction in the data space. 

Similar to what we had in rejection sampling, we draw a random number from the uniform distribution, i.e., $u_{i-1} \sim U(0, P^*(x_{i-1}))$. 
We consider a \textit{slice} with length. or step size, $\delta$, around the point $u_{i-1}$, as shown in Fig. \ref{figure_Slice_sampling}. Note that the point $u_{i-1}$ can be at any location in the slice, and not necessarily in its middle. The method is also very robust to the length of slice or the step size. As long as the end of slices at the end sides fall above the distribution $P^*(X)$, we continue to concatenate slices as shown in the figure. 

Afterwards, a random number is selected in the range of concatenated slices. If the selected point is above the distribution $P^*(X)$, it is rejected and, at the side of the selected point with respect to $x_{i-1}$, the rest of slices is removed until the end. This ensures purifying the range of concatenated slices. If the selected point is rejected, another random point is selected in the purified range of slices. If it falls again above $P^*(X)$, it gets rejected and the slices are purified again. However, if it falls under $P^*(X)$, it is accepted to be next drawn sample, i.e., $x_i$. This procedure is repeated until we have all $\mathcal{S} = \{x_i\}_{i=1}^n$ samples. 
As expected, this algorithm samples more points from the modes of distribution.

\begin{figure*}[!t]
\centering
\includegraphics[width=6.5in]{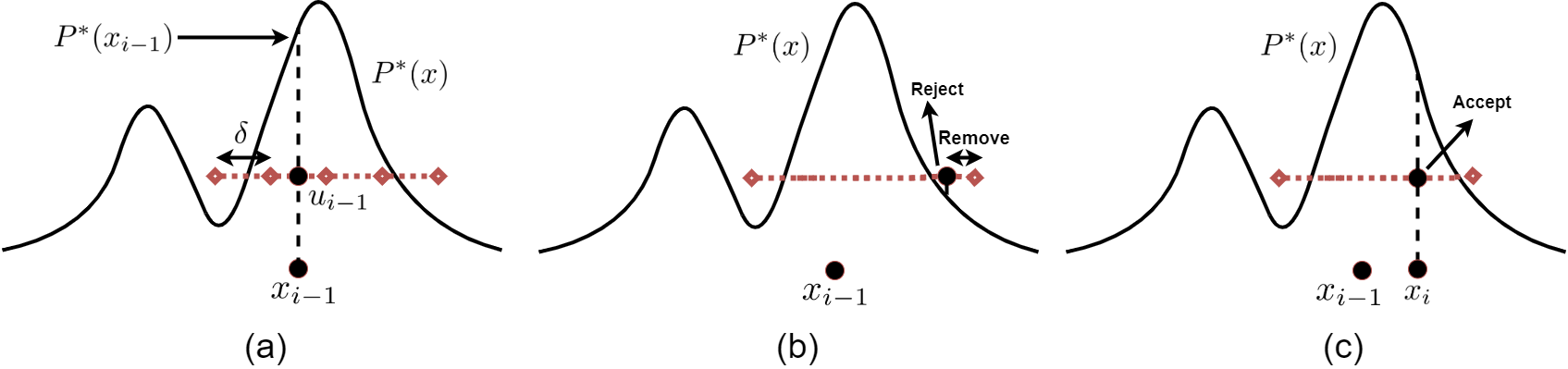}
\caption{The steps in slice sampling algorithm.}
\label{figure_Slice_sampling}
\end{figure*}

\subsection{Efficient Monte Carlo Methods}

\subsubsection{Random Walk Behaviour of Monte Carlo Methods}

Metropolis algorithms have a random walk behaviour \cite{spitzer2013principles}; therefore, they usually progress slowly to explore most of $\textbf{dom}(X)$. For example, we have lengthy random walk in Gibbs sampling, especially when the dimensions are highly correlated \cite{mackay2003information}. 

The following example, taken from \cite{mackay2003information,mackay2014youtube}, can show why the random walk behaviour in Metropolis algorithms is slow. 
Consider a discrete uniform distribution $U\{1, 2, \dots, L\}$. Assume the possible actions for drawing the next sample based on the previous sample is to move $\delta$ steps to left or right, each with probability $0.5$. 
The distance of sample from the sample $T$ iterations before is:
\begin{align}\label{equation_random_walk_Delta_x}
\Delta x = \sum_{i=1}^T s_t,
\end{align}
where $s_t \in \{-1, +1\}$ because the distance of possible values in the distribution $U\{1, 2, \dots, L\}$ is one. 
The variance of this distance is equal to the average, or expected value, of $(\Delta x)^2$ because of the quadratic characteristic of variance. Hence, we have:
\begin{align}
\mathbb{V}\text{ar}(\Delta x) = \langle (\Delta x)^2 \rangle \overset{(\ref{equation_random_walk_Delta_x})}{=} \sum_{i=1}^T \langle S_t^2 \rangle \overset{(a)}{=} T,
\end{align}
where $(a)$ is because $s_t \in \{-1, +1\}$ so $s_t^2 = 1$. 
In order to explore all $\textbf{dom}(X)$ which is $\{1, 2, \dots, L\}$, with the step $\delta$, we want the standard deviation of distance to be $L/\delta$. Therefore, the variance of distance should be $(L/\delta)^2$:
\begin{align}\label{equation_random_walk_complexity}
\sum_{i=1}^T \langle S_t^2 \rangle = T \overset{\text{set}}{=} \big(\frac{L}{\delta}\big)^2.
\end{align}
Hence, the run-time complexity of random walk for exploring the data with range $L$ is quadratic with respect to $L$. 
In other words, we need to draw at least $(L/\delta)^2$ samples to expect to face a fresh independent sample. 

Assume data have $r$ dimensions along which deviation of data, denoted by $\ell$, is roughly small. However, the rest of $d-r$ dimensions have large deviation, denoted by $L$. We have $\ell \ll L$. 
In other words, we have a $r$-dimensional subspace of data \cite{ghojogh2019feature}.
The probability that we accept the new proposal is proportional to the volume of $r$-dimensional hyper-sphere with radius $\ell$ to the volume of $r$-dimensional proposal hyper-sphere with radius/step $\delta$:
\begin{align}
p_\text{accept} = \frac{\ell^r}{\delta^r} = \big(\frac{\ell}{\delta}\big)^r,
\end{align}
because our restriction on acceptance of proposed move in distribution is $\ell$ and not $L$. There is a trade-off here. If we take large step size, i.e. $\delta \gg \ell$, the probability of acceptance becomes very small and we will reject many proposals; therefore, the exploration pacing of algorithm gets very slow. In contrary, if we take $\delta \ll \ell$, the step size gets very small and our exploration pacing gets slow but because of small step size rather than for many rejections. Hence, both very small and very large step sizes are bad choices. A good choice is $\delta \approx \ell$ to both have large enough step size and accept proposals with high probability.   

As Eq. (\ref{equation_random_walk_complexity}) indicates, the time complexity of Metropolis algorithms is quadratic. In the following, we introduce more efficient Metropolis algorithms which explore $\textbf{dom}(X)$ much faster. 

\subsubsection{Hamiltonian (Hybrid) Monte Carlo}

Hamiltonian Monte Carlo (HMC), also called Hybrid Monte Carlo (HMC) \cite{duane1987hybrid}, is used for faster sampling from a distribution compared to Metropolis algorithm. The drawn samples explore the range of data, $\textbf{dom}(X)$, faster. 

Many of the real-world distributions exist in the exponential distribution family \cite{andersen1970sufficiency}. In many physical models, we can model the system with a Boltzmann distribution \cite{cipra1987introduction,mccoy2014two}:
\begin{align}
P^*(X) = \frac{e^{-E(X)}}{Z},
\end{align}
where $Z$ is the normalizing factor or the so-called partition function and $E(X)$ is the energy term. 
In HMC, we augment the state space $X \in \mathbb{R}^d$ with momentum variables $p \in \mathbb{R}^d$ and define the Hamiltonian as:
\begin{align}
\mathbb{R} \ni H(x,p) := E(x) + K(p).
\end{align}
Hence, the distribution is changed to:
\begin{align}\label{equation_P_Boltzmann_with_Hamiltonian}
P^*(X,p) = \frac{e^{-H(X,p)}}{Z} = \frac{1}{Z} e^{-E(X)} e^{-K(p)},
\end{align}
which is separable; therefore, marginalization over $x$ or $p$ can discard the other one. 

HMC makes use of the Newton's law \cite{scheck2010mechanics} because it is very related to physical models. According to the Newton's law, we have \cite{mackay2003information}:
\begin{align}
&\dot{x} = p, \\
&\dot{p} = -\frac{\partial E(x)}{\partial x}, 
\end{align}
where the dot above the variable denotes gradient. 

\begin{definition}
HMC, whose procedure is shown in Algorithm \ref{algorithm_Hamiltonian_MC}, contains three steps iteratively. At every iteration, first, it randomizes the momentum by sampling from marginalization of Eq. (\ref{equation_P_Boltzmann_with_Hamiltonian}) over momentum. 
The sampled momentum is accepted with probability one, as done in Gibbs sampling. 
Then, it applies the Newton's law using leapfrog steps, as described in Algorithm \ref{algorithm_Hamiltonian_MC}, to propose a new sample $x_i$.  
Finally, it decides to accept or reject the newly proposed sample $x_i$ based on change in energy. In physical models, we tend to move toward less energy \cite{cipra1987introduction,mccoy2014two}. Hence, if energy has reduced by the new sample, it is accepted definitely. Otherwise, we accept the proposal with some probability $e^{-\Delta H}$. This behaviour is like the Metropolis algorithm. 
\end{definition}

This method is named hybrid MC because it has a hybrid of behaviours of Gibbs sampling and Metropolis algorithm in sampling the momentum and proposed sample, respectively. 
Note that the randomization of momentum makes the HMC algorithm very fast to explore $\textbf{dom}(X)$. It can be shown that the time complexity if HMC is linear with respect to the range of data, i.e., $\mathcal{O}(L/\delta)$ (cf. Eq. (\ref{equation_random_walk_complexity}) for comparison) \cite{mackay2003information}. 

\SetAlCapSkip{0.5em}
\IncMargin{0.8em}
\begin{algorithm2e}[!t]
\DontPrintSemicolon
    \textbf{Input:} $P^*(X), Q(X_{i+1}; X_{i}), c$\;
    \textbf{Output:} $\mathcal{S} = \{x_i\}_{i=1}^n \sim P^*(X)$\;
    $\mathcal{S} \gets \varnothing$\;
    $x_1 \gets $ a random vector in $\textbf{dom}(X)$\;
    $g_1 \gets \frac{\partial E(x_1)}{\partial x}$\;
    \For{sample index $i$ from $1$ to $n$}{
        // randomize the momentum $p_i$:\;
        $p_i \sim \frac{e^{-K(p)}}{Z} = \mathcal{N}(0,1)$\;
        $K(p_i) \gets \frac{p_i^\top p_i}{2}$\;
        $H \gets E(x_i) + K(p_i)$\;
        // Newton's law (with leapfrog steps):\;
        $x_{i,\text{new}} \gets x_i$\;
        $g_{i,\text{new}} \gets g_i$\;
        \For{$t$ from $1$ to $T$}{
            $p_i \gets p_i - \eta\, \frac{g_{i,\text{new}}}{2}$\;
            $x_{i,\text{new}} \gets x_{i,\text{new}} + \eta\, p_i$\;
            $g_{i,\text{new}} \gets \frac{\partial E(x_{i,\text{new}})}{\partial x}$\;
            $p_i \gets p_i - \eta\, \frac{g_{i,\text{new}}}{2}$\;
        }
        // accept or reject based on energy change:\;
        $K(p_i) \gets \frac{p_i^\top p_i}{2}$\;
        $H_\text{new} \gets E(x_{i,\text{new}}) + K(p_i)$\;
        $\Delta H \gets H_\text{new} - H$\;
        \uIf{$\Delta H < 0$}{
            Accept $x_i$: $\mathcal{S} \gets \mathcal{S} \cup \{x_{i,\text{new}}\}$\; 
        }\uElseIf{$u_i \sim U(0,1) < e^{-\Delta H}$}{
            Accept $x_{i,\text{new}}$: $\mathcal{S} \gets \mathcal{S} \cup \{x_{i,\text{new}}\}$\; 
        }
        \Else{
            Reject $x_{i,\text{new}}$: $i \gets i - 1$\;
        }
        \If{$x_i$ is Accepted}{
            $g_i \gets g_{i,\text{new}}$\;
            $x_i \gets x_{i,\text{new}}$\;
        }
    }
\caption{Hamiltonian (or hybrid) Monte Carlo sampling}\label{algorithm_Hamiltonian_MC}
\end{algorithm2e}
\DecMargin{0.8em}

\subsubsection{Overrelaxation for Gibbs Sampling}

As mentioned before, Gibbs sampling has a very slow random walk behaviour, especially when the dimensions of data are very correlated \cite{mackay2003information}. Some methods, named overrelaxation, are proposed for accelerating the pacing of Gibbs sampling in exploring $\textbf{dom}(X)$ so that the jumps between samples get larger.
In the following, we introduce two methods for overrelaxation. 

\textbf{Adler's overrelaxation:}
\textit{Adler's overrelaxation} \cite{adler1981over} is for a special case where the conditional distributions in Gibbs sampling are all Gaussian distributions.
Its main idea is that, in contrast to Gibbs sampling where $x_{i}^{(j)}$ is independent of $x_{i-1}^{(j)}$, we take $x_{i}^{(j)}$ to be at the \textit{opposite} location of the conditional distribution $P^*(X^{(j)}\, |\, X^{(1)}, \dots, X^{(j-1)}, X^{(j+1)}, \dots, X^{(d)})$ with respect to its expected value (or mean). This also reminds us of the concept of opposition-based learning \cite{tizhoosh2005opposition}. It can be shown that the opposite location of $x_{i-1}^{(j)}$ with respect to the mean of the conditional distribution is:
\begin{align}
x_i^{(j)} \gets \mu + \alpha (x_{i-1}^{(j)} - \mu) + (1 - \alpha^2)^{0.5} \sigma \nu,
\end{align}
where $\nu \sim \mathcal{N}(0,1)$ and $\alpha \in [-1, 1]$ is a parameter which is usually negative (if it is positive, the method is called underrelaxation) \cite{mackay2003information}.
The procedure for Gibbs sampling with Alder's overrelaxation is shown in Algorithm \ref{algorithm_Gibbs_sampling_Adler_overrelaxation}. 

\SetAlCapSkip{0.5em}
\IncMargin{0.8em}
\begin{algorithm2e}[!t]
\DontPrintSemicolon
    \textbf{Input:} $P^*(X), Q(X_{i+1}; X_{i}), c$\;
    \textbf{Output:} $\mathcal{S} = \{x_i\}_{i=1}^n \sim P^*(X)$\;
    $\mathcal{S} \gets \varnothing$\;
    $x_0 \gets $ random $d$-dimensional vector in $\textbf{dom}(X)$\;
    \For{sample index $i$ from $1$ to $n+t_\text{burnIn}$}{
        $x_i \gets x_{i-1}$\;
        \For{dimension $j$ from $1$ to $d$}{
            $x_i^{(j)} \gets \mu + \alpha (x_{i-1}^{(j)} - \mu) + (1 - \alpha^2)^{0.5} \sigma \nu$\;
            $x_i \gets [x_i^{(1)}, x_i^{(2)}, \dots, x_i^{(d)} ]^\top$\;
        }
        \If{$i \geq t_\text{burnIn}$}{
            $\mathcal{S} \gets \mathcal{S} \cup \{x_i\}$\;
        }
    }
\caption{Gibbs sampling with Adler's overrelaxation}\label{algorithm_Gibbs_sampling_Adler_overrelaxation}
\end{algorithm2e}
\DecMargin{0.8em}

\textbf{Ordered overrelaxation:}
In \textit{ordered overrelaxation}
\cite{neal1998suppressing}, rather than sampling directly from the opposite location of $x_{i-1}^{(j)}$ with respect to the mean of the conditional distribution, we draw $K-1$ other sample from the conditional distribution, to have a total of $K$ samples including $x_{i-1}^{(j)}$ itself. A good value for $K$ is $20$ \cite{mackay2014youtube}. Then, we see what order statistic $x_{i-1}^{(j)}$ has, i.e., if we sort the $K$ samples, which index it gets. If its order statistic is $k$, we take $k$ samples from the end (opposite direction), i.e. the sample with sorting index $(K-k)$, to be $x_{i}^{(j)}$. The opposite direction behaviour reminds us of the concept of opposition-based learning \cite{tizhoosh2005opposition}, again.

\section{Summary of Characteristics, Discussion, and Conclusion}\label{section_conclusion}

In this section, we briefly summarize sampling algorithms and review their pros and cons. Sampling algorithms divide into two main categories, i.e., survey sampling and sampling from distributions using Monte Carlo methods. In survey sampling, we have a set of data points or vectors and we sample from these points. However, in Monte Carlo methods, we sample from a distribution of data. 

\subsection{Summary of Survey Sampling}

There are various survey sampling methods such as SRS, bootstrapping, stratified sampling, cluster sampling, multistage sampling, network sampling, and snowball sampling. 
SRS is sampling without replacement. Bootstrapping, however, is sampling with replacement. If data can be divided into several strata, stratified sampling draws samples by SRS from each stratum. Likewise, if data can be divided into several clusters, cluster sampling samples clusters as blocks of data, in the cluster level, using SRS. We showed that if the clusters or strata of data are significantly different from each other (i.e., if we have large between-variance and small within-variance), stratified sampling is useful. In contrary, if the clusters or strata of data are mostly similar to each other (i.e., if we have small between-variance and large within-variance), cluster sampling is better to use.
Stratified sampling definitely makes the variance of estimation less than (or equal to) SRS; so it is always better to use stratified sampling rather than SRS, even by dividing data into some not necessarily perfect strata. Cluster sampling may or may not reduce the variance of estimation less than SRS. 
Multistage sampling draws samples stage-wise and can be used to combine different survey sampling methods. Network sampling is a family of methods for sampling sub-networks from a graph or network. A special case of network sampling is snowball sampling which draws some initial samples; then, gives choice of sampling to the selected samples to draw other samples based on their own decision. 

\subsection{Summary of Monte Carlo Methods}

Sampling from distribution of data is usually performed using Monte Carlo methods. Monte Carlo approximation is used for approximating expectation or probability of a function of data over the distribution. Monte Carlo methods can be divided into simple Monte Carlo methods and MCMC. It is noteworthy that Monte Carlo methods are iterative. 

In simple methods, every iteration is independent of the previous iteration because iterations are performed blindly using a simple-to-sample distribution. Some simple Monte Carlo methods are importance sampling and rejection sampling. Importance sampling is used to approximate the expectation of a function of data over the complicated distribution using another simple-to-sample distribution. Rejection sampling is for sampling from complicated distributions using a simple-to-sample upper-bound distribution. 

In MCMC, every iteration is dependent on the previous iteration so sampling is not blind but it has the memory of Markov property. Some MCMC methods are Metropolis algorithm, Metropolis-Hastings algorithm, Gibbs sampling, and slice sampling. Metropolis algorithm draws the next sample using a simple-to-sample distribution whose mean is the previous sample. This proposal function is symmetric in the Metropolis algorithm. By modifying the probability of acceptance of proposal, the Metropolis-Hastings algorithm generalizes the Metropolis algorithm by relaxing the symmetric restriction on the proposal function. Gibbs sampling draws samples using conditional distributions of every coordinate conditioned on the rest of coordinates. Gibbs sampling can be considered as a special case of Metropolis-Hastings algorithm with probability one. 
One of the problems of Metropolis algorithms is choosing an appropriate step size. In contrary, slice sampling is a MCMC method which is robust to the step size. 
Slice sampling considers slices on the sides of previous sample and draws samples in the range of those slices.  

Another issue with Monte Carlo methods is their slow random walk behaviour. Hamiltonian or hybrid Monte Carlo is a Monte Carlo method which is faster for exploration of range of data. Moreover, overrelaxation methods, such as Adler's overrelaxation and ordered overrelaxation, can be used to make Gibbs sampling faster to explore the range of data, especially when the dimensions of data are highly correlated. 

\subsection{Some Other Not Covered Sampling Methods}

For the sake of brevity, we did not cover the Thompson sampling \cite{thompson1933likelihood,russo2018tutorial}, which is useful in reinforcement learning \cite{sutton2018reinforcement}. 
Moreover, exact sampling was not covered. In short, exact sampling is a family of methods which start from some iteration before some state with different initial states. If those multiple runs with different initializations converge to the same state in the time span, we are done. Otherwise, we go back further in the past and start the processes. We do this until all the processes with different initializations converge to the same desired state. 
Exact sampling can be for discrete \cite{burr1955calculation} and continuous \cite{murdoch1998exact} state spaces.

\section*{Acknowledgment}

The authors hugely thank Prof. Mu Zhu \cite{zhu2017lectureSurveySampling,zhu2017lectureMCMC}, Prof. David McKay \cite{mackay2003information,mackay2014youtube}, Prof. Mehdi Molkaraie, and Prof. Kevin Granville whose courses partly covered the materials mentioned in this tutorial paper.

\appendix

\section{Proofs for Section \ref{section_background}}\label{section_appendix_background}

\subsection{Proof for Eq. (\ref{equation_variance_2})}

\begin{align}
\mathbb{V}\text{ar}(\widehat{X}) &= \mathbb{E}\big( \widehat{X}^2 + (\mathbb{E}(\widehat{X}))^2 - 2\widehat{X}\mathbb{E}(\widehat{X}) \big) \nonumber \\
&\overset{(a)}{=} \mathbb{E}(\widehat{X}^2) + (\mathbb{E}(\widehat{X}))^2 - 2\mathbb{E}(\widehat{X})\mathbb{E}(\widehat{X}) \nonumber \\
&= \mathbb{E}(\widehat{X}^2) - (\mathbb{E}(\widehat{X}))^2, \nonumber
\end{align}
where $(a)$ is because expectation is a linear operator and $\mathbb{E}(\widehat{X})$ is not a random variable.

\subsection{Proof for Eq. (\ref{equation_relation_MSE_variance_bias})}

\begin{align*}
\text{MSE}(\widehat{X}) &= \mathbb{E}\big((\widehat{X} - X)^2\big) \nonumber \\
&= \mathbb{E}\big((\widehat{X} - \mathbb{E}(\widehat{X}) + \mathbb{E}(\widehat{X}) - X)^2\big) \nonumber 
\end{align*}
\begin{align*}
&\quad\quad= \mathbb{E}\big((\widehat{X} - \mathbb{E}(\widehat{X}))^2 + (\mathbb{E}(\widehat{X}) - X)^2 \nonumber \\
&\quad\quad~~~~ + 2 (\widehat{X} - \mathbb{E}(\widehat{X})) (\mathbb{E}(\widehat{X}) - X) \big) \nonumber 
\end{align*}
\begin{align*}
&\overset{(a)}{=} \mathbb{E}\big( (\widehat{X} - \mathbb{E}(\widehat{X}))^2 \big) + (\mathbb{E}(\widehat{X}) - X)^2 \nonumber \\
&~~~~ + 2 \underbrace{(\mathbb{E}(\widehat{X}) - \mathbb{E}(\widehat{X}))}_{0} (\mathbb{E}(\widehat{X}) - X) \nonumber \\
&\overset{(b)}{=} \mathbb{V}\text{ar}(\widehat{X}) + (\mathbb{B}\text{ias}(\widehat{X}))^2, \nonumber
\end{align*}
where $(a)$ is because expectation is a linear operator and $X$ and $\mathbb{E}(\widehat{X})$ are not random, and $(b)$ is because of Eqs. (\ref{equation_variance}) and (\ref{equation_bias}).

\subsection{Proof for Eq. (\ref{equation_variance_of_two_variables})}

\begin{align}
&\mathbb{V}\text{ar}(a\widehat{X} + b\widehat{Y}) \overset{(\ref{equation_variance_2})}{=} \mathbb{E}\big((a\widehat{X} + b\widehat{Y})^2\big) - \big(\mathbb{E}(a\widehat{X} + b\widehat{Y})\big)^2 \nonumber \\
&\overset{(a)}{=} a^2\, \mathbb{E}(\widehat{X}^2) + b^2\, \mathbb{E}(\widehat{Y}^2) + 2ab\, \mathbb{E}(\widehat{X}\widehat{Y}) \nonumber \\
&~~~~ - a^2\, (\mathbb{E}(\widehat{X}))^2 - b^2\, (\mathbb{E}(\widehat{Y}))^2 - 2ab\, \mathbb{E}(\widehat{Y})\mathbb{E}(\widehat{Y}) \nonumber \\
&\overset{(\ref{equation_variance_2})}{=} a^2\, \mathbb{V}\text{ar}(\widehat{X}) + b^2\, \mathbb{V}\text{ar}(\widehat{X}) + 2ab\, \mathbb{C}\text{ov}(\widehat{X},\widehat{Y}), \nonumber
\end{align}
where $(a)$ is because of linearity of expectation and the $\mathbb{C}\text{ov}(\widehat{X},\widehat{Y})$ is covariance defined in Eq. (\ref{equation_covariance}). 

\subsection{Proof for Eq. (\ref{equation_expectation_independent})}

\begin{align}
&\mathbb{E}(\widehat{X}\widehat{Y}) \overset{(a)}{=} \int\!\!\! \int \widehat{x} \widehat{y} f(\widehat{x}, \widehat{y}) d\widehat{x} d\widehat{y} \overset{\indep}{=} \int\!\!\! \int \widehat{x} \widehat{y} f(\widehat{x}) f(\widehat{y}) d\widehat{x} d\widehat{y} \nonumber \\
&= \int \widehat{y} f(\widehat{y}) \underbrace{\int \widehat{x} f(\widehat{x}) d\widehat{x}}_{\mathbb{E}(\widehat{X})} d\widehat{y} = \mathbb{E}(\widehat{X}) \underbrace{\int \widehat{y} f(\widehat{y}) d\widehat{y}}_{\mathbb{E}(\widehat{Y})}  \nonumber \\
&= \mathbb{E}(\widehat{X})\, \mathbb{E}(\widehat{Y}) \implies \mathbb{C}\text{ov}(\widehat{X},\widehat{Y}) = 0, \nonumber
\end{align}
where $(a)$ is according to definition of expectation. 

\subsection{Proof for Lemma \ref{lemma_variance_estimate_restate}}

\begin{align*}
\sigma^2 &= \frac{1}{N} \sum_{j=1}^N (x_j - \mu)^2 = \frac{1}{N} \sum_{j=1}^N (x_j^2 - 2 \mu x_j + \mu^2) \\
&= \frac{1}{N} \Big(\sum_{j=1}^N x_j^2 - 2 \mu \sum_{j=1}^N x_j + \mu^2 \sum_{j=1}^N 1\Big) 
\end{align*}
\begin{align*}
&\overset{(\ref{equation_mean_estimate})}{=}  \frac{1}{N} \Big(\sum_{j=1}^N x_j^2 - 2 \mu N \mu + \mu^2 N\Big) \\
&= \frac{1}{N} \Big(\sum_{j=1}^N x_j^2 - \mu^2 N\Big) = \frac{1}{N} \sum_{j=1}^N x_j^2 - \mu^2. \quad \text{Q.E.D.}
\end{align*}

\subsection{Proof for Lemma \ref{equation_variance_of_mean}}

\begin{align*}
\mathbb{V}\text{ar}(\mu) &\overset{(\ref{equation_mean_estimate})}{=} \mathbb{V}\text{ar}( \frac{1}{N} \sum_{j=1}^N x_j ) \\
&\overset{(\ref{equation_variance_multiple_independent})}{=} \frac{1}{N^2} \mathbb{V}\text{ar}(x_1) + \dots + \frac{1}{N^2} \mathbb{V}\text{ar}(x_N) \\
&\overset{(a)}{=} \frac{1}{N^2} \times N \times \mathbb{V}\text{ar}(X) = \frac{1}{N} \mathbb{V}\text{ar}(X) = \frac{1}{N}\, \sigma^2,
\end{align*}
where $(a)$ is because the items of sample are independent and identically distributed (iid). Q.E.D.

\subsection{Proof for Proposition \ref{proposition_variance_unbiased}}

According to Lemma \ref{lemma_variance_estimate_restate} and by comparing Eqs. (\ref{equation_variance_estimate}) and (\ref{equation_variance_estimate_2}), we have (if we multiply the sides by $N$):
\begin{align}
\sum_{j=1}^N (x_j - \mu)^2 = \sum_{j=1}^N x_j^2 - N \mu^2.
\end{align}
Hence, we have:
\begin{align}
\mathbb{E}\Big(\sum_{j=1}^N (x_j - \mu)^2\Big) &= \mathbb{E}\Big(\sum_{j=1}^N x_j^2\Big) - N \mathbb{E}(\mu^2) \nonumber \\
&\overset{(a)}{=} \sum_{j=1}^N \mathbb{E}(x_j^2) - N \mathbb{E}(\mu^2), \label{equation_proposition1_middleEq_1}
\end{align}
where $(a)$ is because expectation is a linear operator. 
According to Eq. (\ref{equation_variance_2}), we have:
\begin{align*}
&\mathbb{E}(x_j^2) = \mathbb{V}\text{ar}(x_j) + \big(\mathbb{E}(x_j)\big)^2 = \sigma^2 + \mu^2, \\
&\mathbb{E}(\mu^2) = \mathbb{V}\text{ar}(\mu) + \big(\mathbb{E}(\mu)\big)^2 \overset{(\ref{equation_variance_of_mean})}{=} \frac{\sigma^2}{N} + \mu^2.
\end{align*}
Plugging these into Eq. (\ref{equation_proposition1_middleEq_1}) gives:
\begin{align*}
&\mathbb{E}\Big(\sum_{j=1}^N (x_j - \mu)^2\Big) = (\sigma^2 + \mu^2) \sum_{j=1}^N (1) - N (\frac{\sigma^2}{N} + \mu^2) \\
&= N \sigma^2 + N \mu^2 - \sigma^2 - N \mu^2 = (N-1)\, \sigma^2. 
\end{align*}
Hence, the expectation of Eq. (\ref{equation_variance_unbiased}) is:
\begin{align*}
\mathbb{E}(\sigma^2) &= \frac{1}{N-1} \mathbb{E}\Big(\sum_{j=1}^N (x_j - \mu)^2\Big) \\
&= \frac{1}{N-1} (N-1)\, \sigma^2 = \sigma^2. \quad \text{Q.E.D.}
\end{align*}

\subsection{Proof for Proposition \ref{proposition_HT_unbiased}}

\begin{align*}
&\widehat{\theta}_\text{HT} \overset{(\ref{equation_HT_estimator})}{=} \sum_{j \in \mathcal{S}} \frac{h(x_j)}{\pi_j} = \sum_{j = 1}^{N} \frac{h(x_j)}{\pi_j} \mathbb{I}_j\\
&\mathbb{E}(\widehat{\theta}_\text{HT}) \overset{(a)}{=} \sum_{j = 1}^{N} \frac{h(x_j)}{\pi_j} \mathbb{E}(\mathbb{I}_j) \overset{(b)}{=} \sum_{j = 1}^{N} \frac{h(x_j)}{\pi_j} \pi_j \\
&~~~~~~~~~~~~= \sum_{j = 1}^{N} h(x_j) \overset{(\ref{equation_estimator_population_quantity})}{=} \theta,
\end{align*}
where $(a)$ is because expectation is linear and $(b)$ is because $\mathbb{E}(\mathbb{I}_j) = (0 \times (1- \pi_j)) + (1 \times \pi_j) = \pi_j$. Q.E.D.

\section{Proofs for Section \ref{section_survey_sampling}}\label{section_appendix_survey_sampling}

\subsection{Proof for Proposition \ref{proposition_expectation_variance_of_mean_SRS}}

This proof is based on \cite{zhu2017lectureSurveySampling}. 

The expectation of mean of sample by SRS is:
\begin{align*}
\mathbb{E}(\widehat{\mu}) &\overset{(\ref{equation_mean_SRS})}{=} \mathbb{E}(\frac{1}{n} \sum_{j=1}^N x_j\, \mathbb{I}_j) \overset{(a)}{=} \frac{1}{n} \sum_{j=1}^N x_j\, \mathbb{E}(\mathbb{I}_j) \\
&\overset{(b)}{=} \frac{1}{n} \sum_{j=1}^N x_j\, \pi_j = \frac{1}{n} \sum_{j=1}^N x_j\, \frac{n}{N} = \frac{1}{N} \sum_{j=1}^N x_j = \mu.
\end{align*}
where $(a)$ is because expectation is linear and $(b)$ is because:
\begin{align*}
\mathbb{E}(\mathbb{I}_j) &= (0 \times \mathbb{P}(j \not\in \mathcal{S})) + (1 \times \mathbb{P}(j \in \mathcal{S})) \\
&= \mathbb{P}(j \in \mathcal{S}) = \pi_j
\end{align*}
Note that $\mathbb{E}(\widehat{\mu}) = \mu$, proved above, was expected because according to Proposition \ref{proposition_HT_unbiased}, the mean of sample in SRS is an \textit{unbiased} estimate of the mean of whole data. Hence, according to Definition \ref{definition_unbiased_estimator}, the expectation is the variable itself. 

\begin{lemma}
Variance can be restated as:
\begin{align}\label{equation_variance_restate_middleOfProof}
&\sigma^2 = \frac{1}{N(N-1)} \Big((N-1) \sum_{j=1}^N x_j^2 - \sum_{j \neq \ell} x_j x_\ell\Big).
\end{align}
\end{lemma}
\begin{proof}
\begin{align*}
&\sigma^2 \overset{(\ref{equation_variance_unbiased})}{=} \frac{1}{N-1} \sum_{j=1}^N (x_j - \mu)^2  \\
&= \frac{1}{N-1} \sum_{j=1}^N (x_j^2 + \mu^2 - 2 \mu x_j) \\
&= \frac{1}{N-1} \Big(\sum_{j=1}^N x_j^2 + \mu^2 \sum_{j=1}^N 1 - 2 \mu \sum_{j=1}^N x_j\Big) \\
&\overset{(\ref{equation_mean_estimate})}{=} \frac{1}{N-1} \Big(\sum_{j=1}^N x_j^2 + \mu^2 N - 2 \mu N \mu\Big)  \\
&= \frac{1}{N-1} \Big(\sum_{j=1}^N x_j^2 - N \mu^2 \Big).
\end{align*}
On the other hand, the squared mean can be restated as:
\begin{align*}
\mu^2 \overset{(\ref{equation_mean_estimate})}{=} \Big(\frac{1}{N} \sum_{j=1}^N x_j\Big)^2 = \frac{1}{N^2} \Big( \sum_{j=1}^N x_j^2 + \sum_{j \neq \ell} x_j x_\ell \Big).
\end{align*}
Therefore:
\begin{align*}
\sigma^2 &= \frac{1}{N-1} \Big(\sum_{j=1}^N x_j^2 - \frac{1}{N} \Big( \sum_{j=1}^N x_j^2 + \sum_{j \neq \ell} x_j x_\ell \Big) \Big) \\
&= \frac{1}{N(N-1)} \Big((N-1) \sum_{j=1}^N x_j^2 - \sum_{j \neq \ell} x_j x_\ell \Big).
\end{align*}
\end{proof}

The variance of mean of sample by SRS is:
\begin{align*}
&\mathbb{V}\text{ar}(\widehat{\mu}) \overset{(\ref{equation_variance_multiple})}{=} \frac{1}{n^2} \Big( \sum_{j=1}^N x_j^2 \mathbb{V}\text{ar}(\mathbb{I}_j) + \sum_{j \neq \ell} x_j x_\ell \mathbb{C}\text{ov}(\mathbb{I}_j, \mathbb{I}_\ell) \Big) \\
&\overset{(a)}{=} \frac{1}{n^2} \Big( \sum_{j=1}^N x_j^2 \pi_j (1-\pi_j) + \sum_{j \neq \ell} x_j x_\ell (\pi_{j\ell}-\pi_j \pi_\ell) \Big) 
\end{align*}
\begin{align*}
&\overset{(b)}{=} \frac{1}{n^2} \Big( \sum_{j=1}^N x_j^2 \frac{n}{N} (1-\frac{n}{N}) \\
&~~~~~~~~~~~~~~~~~~~~~~~~~+ \sum_{j \neq \ell} x_j x_\ell (\frac{n}{N}\frac{n-1}{N-1}-\frac{n}{N}\frac{n}{N}) \Big) \\
&= \frac{1}{n^2} \frac{n}{N} \Big( \sum_{j=1}^N x_j^2 (\frac{N-n}{N}) + \sum_{j \neq \ell} x_j x_\ell \frac{n-N}{N(N-1)} \Big) 
\end{align*}
\begin{align*}
&= \frac{1}{n^2} \frac{n}{N} \Big( \sum_{j=1}^N x_j^2 (\frac{N-n}{N}) + \sum_{j \neq \ell} x_j x_\ell \frac{n-N}{N(N-1)} \Big) \\
&= \frac{1}{n^2} \frac{n}{N} \frac{N-n}{N} \frac{1}{N-1} \Big( (N-1) \sum_{j=1}^N x_j^2 - \sum_{j \neq \ell} x_j x_\ell \Big) \\
&\overset{(\ref{equation_variance_restate_middleOfProof})}{=} \frac{1}{n^2} \frac{n}{N} \frac{N-n}{N} \frac{1}{N-1} \Big( N(N-1)\,\sigma^2 \Big) \\
&= \big(1 - \frac{n}{N}\big) \frac{\sigma^2}{n},
\end{align*}
where $(a)$ is because variance of the Bernoulli distribution is Eq. (\ref{equation_Bernoulli_variance}). Also, according to Eq. (\ref{equation_covariance}), we have:
\begin{align*}
&\mathbb{C}\text{ov}(\mathbb{I}_j, \mathbb{I}_\ell) := \mathbb{E}(\mathbb{I}_j \mathbb{I}_\ell) - \mathbb{E}(\mathbb{I}_j)\,\mathbb{E}(\mathbb{I}_\ell) \\
&= \Big[ [0 \times 1 \times \mathbb{P}(j \not \in \mathbb{S} \land \ell \in \mathbb{S})] \\
&+ [1 \times 0 \times \mathbb{P}(j \in \mathbb{S} \land \ell \not \in \mathbb{S})] \\
&+ [0 \times 0 \times \mathbb{P}(j \not \in \mathbb{S} \land \ell \not \in \mathbb{S})] \\
&+ [1 \times 1 \times \mathbb{P}(j \in \mathbb{S} \land \ell \in \mathbb{S})] \Big] \\
&- \Big[ [0 \times \mathbb{P}(j \not \in \mathbb{S})] + [1 \times \mathbb{P}(j \in \mathbb{S})] \Big] \\
&~~~~~~~~\Big[ [0 \times \mathbb{P}(\ell \not \in \mathbb{S})] + [1 \times \mathbb{P}(\ell \in \mathbb{S})] \Big] \\
&= \mathbb{P}(j \in \mathbb{S} \land \ell \in \mathbb{S}) + \mathbb{P}(j \in \mathbb{S}) \mathbb{P}(\ell \in \mathbb{S}) = \pi_{j \ell} - \pi_j \pi_\ell.
\end{align*}
Moreover, $(b)$ is because:
\begin{align*}
&\pi_j = \mathbb{P}(j \in \mathcal{S}) = \frac{n}{N}, \\
&\pi_\ell = \mathbb{P}(\ell \in \mathcal{S}) = \frac{n}{N}, \\
&\pi_{j\ell} = \mathbb{P}(j \in \mathcal{S} \land \ell \in \mathcal{S}) \\
&~~~~~~= \mathbb{P}(j \in \mathcal{S} | \ell \in \mathcal{S})\, \mathbb{P}(\ell \in \mathcal{S}) = \frac{n-1}{N-1} \frac{n}{N}. \\
\end{align*}

\subsection{Proof for Corollary \ref{corollary_SRS_varianceOfMean_proportional}}

\begin{align*}
&\mathbb{V}\text{ar}(\widehat{\mu}) \overset{(\ref{equation_SRS_mean_variance})}{=} \big(1 - \frac{n}{N}\big)\, \frac{\sigma^2}{n} \overset{(\ref{equation_variance_with_strata})}{=} 
\\
&\frac{1}{n} (1 - \frac{n}{N}) \bigg[ \sum_{k=1}^K \frac{N_k - 1}{N - 1} \sigma_k^2 + \sum_{k=1}^K \frac{N_k}{N-1} (\mu_k - \mu)^2 \bigg].
\end{align*}

\subsection{Proof for Proposition \ref{proposition_expectation_variance_of_mean_stratified_sampling}}

\begin{align*}
\mathbb{E}(\widehat{\mu}) &\overset{(\ref{equation_mean_estimate_stratified_sampling})}{=} \mathbb{E}\big(\sum_{k=1}^K \frac{N_k}{N} \widehat{\mu}_k\big) = \sum_{k=1}^K \frac{N_k}{N} \mathbb{E}(\widehat{\mu}_k) \\
&\overset{(\ref{equation_stratifiedSampling_mean_expectation})}{=} \sum_{k=1}^K \frac{N_k}{N} \mu_k = \mu,
\end{align*}
which makes sense because according to Proposition \ref{proposition_HT_unbiased}, the HT estimator is unbiased. 

\begin{align*}
\mathbb{V}\text{ar}(\widehat{\mu}) &\overset{(\ref{equation_mean_estimate_stratified_sampling})}{=} \mathbb{V}\text{ar}\big(\sum_{k=1}^K \frac{N_k}{N} \widehat{\mu}_k\big) \overset{(a)}{=} \sum_{k=1}^K \big(\frac{N_k}{N}\big)^2 \mathbb{V}\text{ar}(\widehat{\mu}_k) \\
&\overset{(\ref{equation_stratifiedSampling_mean_variance})}{=} \sum_{k=1}^K \big(\frac{N_k}{N}\big)^2 \big(1 - \frac{n_k}{N_k}\big)\, \frac{\sigma_k^2}{n},
\end{align*}
where $(a)$ is because of Eq. (\ref{equation_variance_multiple_independent}) where the strata are disjoint and thus independent. Q.E.D.

\subsection{Proof for Eq. (\ref{equation_mean_with_strata}) in Lemma \ref{lemma_mean_variance_with_strata}}

According to Eq. (\ref{equation_mean_estimate}), the actual mean of the $k$-th stratum is:
\begin{align*}
&\mu_k = \frac{1}{N_k} \sum_{j=1}^{N_k} x_{k,j} \implies \sum_{j=1}^{N_k} x_{k,j} = N_k\, \mu_k \nonumber \\
&\therefore\quad\mu = \frac{1}{N} \sum_{k=1}^K \Big[\sum_{j=1}^{N_k} x_{k,j}\Big] = \sum_{k=1}^K \frac{N_k}{N} \mu_k. \quad \text{Q.E.D.}
\end{align*}

\subsection{Proof for Eq. (\ref{equation_variance_with_strata}) in Lemma \ref{lemma_mean_variance_with_strata}}

This proof is based on \cite{zhu2017lectureSurveySampling}. 
According to Eq. (\ref{equation_variance_unbiased}), the total variance of data, with $k$ strata, is:
\begin{align*}
\sigma^2 &= \frac{1}{N-1} \sum_{k=1}^K \sum_{j=1}^{N_k} (x_{k,j} - \mu)^2 \\
&= \frac{1}{N-1} \sum_{k=1}^K \sum_{j=1}^{N_k} (x_{k,j} - \mu_k + \mu_k - \mu)^2 \\
&= \frac{1}{N-1} \Big[ \sum_{k=1}^K \sum_{j=1}^{N_k} (x_{k,j} - \mu_k)^2 + \sum_{k=1}^K \sum_{j=1}^{N_k} (\mu_k - \mu)^2 \\
&~~~~~~~~~~~~ + 2 \sum_{k=1}^K \sum_{j=1}^{N_k} (x_{k,j} - \mu_k) (\mu_k - \mu) \Big].
\end{align*}
As the strata are independent (because they are disjoint), the third term is zero. 
The first term is $\sum_{k=1}^K (N_k - 1) \sigma_k^2$, according to Eq. (\ref{equation_variance_unbiased_stratum}). The second term is $\sum_{k=1}^K \sum_{j=1}^{N_k} (\mu_k - \mu)^2 = \sum_{k=1}^K (\mu_k - \mu)^2 \sum_{j=1}^{N_k} 1 = \sum_{j=1}^{N_k} N_k (\mu_k - \mu)^2$. 
Hence:
\begin{align*}
\sigma^2 &= \frac{1}{N-1} \Big[ \sum_{k=1}^K (N_k - 1) \sigma_k^2 + \sum_{j=1}^{N_k} N_k (\mu_k - \mu)^2 \Big].
\end{align*}

\subsection{Proof for Corollary \ref{corollary_stratified_sampling_varianceOfMean_proportional}}

According to Eq. (\ref{equation_stratifiedSampling_mean_variance_total}), we have:
\begin{align*}
\mathbb{V}\text{ar}(\widehat{\mu}) &= \sum_{k=1}^K \big(\frac{N_k}{N}\big)^2 \big(1 - \frac{n_k}{N_k}\big)\, \frac{\sigma_k^2}{n_k} \\
& \overset{(\ref{equation_proportional_allocation})}{=} \sum_{k=1}^K \big(\frac{N_k}{N}\big)^2 \big(1 - \frac{n N_k}{N N_k}\big)\, \frac{\sigma_k^2 N}{n N_k} \\
&= \frac{1}{n} (1 - \frac{n}{N}) \sum_{k=1}^K \big(\frac{N_k}{N}\big)^2 \big(\frac{N}{N_k}\big)\, \sigma_k^2 \\
&= \frac{1}{n} (1 - \frac{n}{N}) \sum_{k=1}^K \big(\frac{N_k}{N}\big)\, \sigma_k^2. \quad~~~~ \text{Q.E.D.}
\end{align*}

\subsection{Proof for Proposition \ref{proposition_expectation_variance_of_mean_cluster_sampling}}

According to Proposition \ref{proposition_HT_unbiased}, the HT estimator is unbiased; hence:
\begin{align*}
\mathbb{E}(\widehat{\mu}) = \mu.
\end{align*}
Moreover:
\begin{align*}
\mathbb{V}\text{ar}(\widehat{\mu}) &\overset{(\ref{equation_mean_estimate_cluster_sampling_2})}{=} \mathbb{V}\text{ar}(\frac{K}{N}\, \widehat{\mu}_*) =
\frac{K^2}{N^2} \mathbb{V}\text{ar}(\widehat{\mu}_*)
\\
&\overset{(\ref{equation_SRS_mean_variance})}{=} \frac{K^2}{N^2} (1 - \frac{c}{K}) \frac{\sigma_*^2}{c},
\end{align*}
where it is noticed that $\widehat{\mu}_*$ is the mean of cluster level with SRS sampling approach according to the definition of cluster sampling (see Definition \ref{definition_cluster_sampling}). Q.E.D.

\subsection{Proof for Corollary \ref{corollary_cluster_sampling_varianceOfMean_equalClusters}}

According to Corollary \ref{corollary_mean_variance_with_clusters}, we have:
\begin{align*}
&\mu_k = \frac{1}{N_k} \sum_{j=1}^{N_k} x_{k,j} \implies \sum_{j=1}^{N_k} x_{k,j} = N_k\, \mu_k, \\
&\therefore~~ \mu = \frac{1}{N} \sum_{k=1}^K \sum_{j=1}^{N_k} x_{k,j} = \frac{1}{N} \sum_{k=1}^K N_k \mu_k, \\
&\implies \sum_{k=1}^K N_k \mu_k = N \mu
\end{align*}
Hence:
\begin{align*}
\frac{1}{K} \sum_{k'=1}^K \tau_{k'} = \frac{1}{K} \sum_{k'=1}^K N_{k'} \mu_{k'} = \frac{1}{K} N \mu \overset{(\ref{equation_euqal_cluster_size})}{=} L\, \mu.
\end{align*}
Also:
\begin{align*}
\tau_k = N_k \mu_k \overset{(\ref{equation_euqal_cluster_size})}{=} L\, \mu_k.
\end{align*}
According to Eq. (\ref{equation_clusterSampling_mean_variance_total}), we have:
\begin{align*}
\mathbb{V}\text{ar}(\widehat{\mu}) &= \frac{K^2}{N^2} (1 - \frac{c}{K}) \frac{1}{c} \frac{1}{K-1} \sum_{k=1}^K (\tau_k - \frac{1}{K} \sum_{k'=1}^K \tau_{k'})^2 \\
&\overset{(\ref{equation_euqal_cluster_size})}{=} \frac{1}{L^2} (1 - \frac{c}{K}) \frac{1}{c} \frac{1}{K-1} \sum_{k=1}^K (L \mu_k - L \mu)^2 \\
&= \frac{1}{c} \big(1 - \frac{c}{K}\big) \bigg[ \frac{1}{K-1} \sum_{k=1}^K (\mu_k - \mu)^2 \bigg]. \quad \text{Q.E.D.}
\end{align*}

\section{Proofs for Section \ref{section_Monte_Carlo}}\label{section_appendix_Monte_Carlo}

\subsection{Proof for Proposition \ref{proposition_importance_sampling_expectation}}

Consider the term $\frac{P^*(X)}{Q(X)} h(X)$ in Eq. (\ref{equation_importance_sampling_expectation}). 
The expectation of this term over the simple distribution $Q(X)$ is:
\begin{align*}
\mathbb{E}_{\sim Q(X)}\Big( \frac{P^*(X)}{Q(X)} h(X) \Big) &\overset{(\ref{equation_MC_approximate_expectation_exact})}{=} \int \frac{P^*(x)}{Q(x)} h(x)\, Q(x)\, dx \\
&= \int P^*(x)\, h(x)\, dx,
\end{align*}
which is, up to scale, the desired expectation of $h(x)$ over the complicated distribution $f(x)$, i.e., $\mathbb{E}_{\sim f(x)} (h(x))$. Q.E.D.

\subsection{Proof for Lemma \ref{lemma_stationary_balance_condition}}

Proof is based on \cite{zhu2017lectureMCMC}.
\begin{align*}
\int \mathbb{P}(u)\, A(v; u)\, du &\overset{(\ref{equation_stationary_balance_condition})}{=} \int \mathbb{P}(v)\, A(u; v)\, du \\
&= \mathbb{P}(v) \int A(u; v)\, du \overset{(a)}{=} \mathbb{P}(v),
\end{align*}
where $(a)$ is because the integral is equal to one because sums over all possible transitions from states to the state $v$. 
According to Eq. (\ref{equation_stationary_distribution_sumInward}), the above expression implies the definition of a stationary distribution. Q.E.D.

\subsection{Proof for Proposition \ref{proposition_Metropolis_symmetric}}

According to Eq. (\ref{equation_Metropolis_transition_function}), the Eq. (\ref{equation_stationary_balance_condition}) in Metropolis algorithm becomes:
\begin{align*}
&P^*(x_{i-1})\, A(x_{i}; x_{i-1})  \\
&\overset{(\ref{equation_Metropolis_transition_function})}{=} \min\Big(P^*(x_{i-1})\, Q(x_{i}; x_{i-1}) \frac{P^*(x_i)}{P^*(x_{i-1})}, \\
&~~~~~~~~~~~~~~~~~~~~~~ P^*(x_{i-1})\, Q(x_{i}; x_{i-1})\Big). \\
&= \min\Big( P^*(x_i)\, Q(x_{i}; x_{i-1}), P^*(x_{i-1})\, Q(x_{i}; x_{i-1}) \Big) \\
&\overset{(\ref{equation_Metropolis_symmetric_proposal_function})}{=} \min\Big( P^*(x_i)\, Q(x_{i-1}; x_{i}), P^*(x_{i-1})\, Q(x_{i}; x_{i-1}) \Big) \\
&\overset{(\ref{equation_Metropolis_transition_function})}{=} P^*(x_{i})\, A(x_{i-1}; x_{i}),
\end{align*}
which is a stationary distribution according to Eq. (\ref{equation_stationary_balance_condition}). Q.E.D.

\subsection{Proof for Proposition \ref{proposition_Gibbs_special_case_Metropolis}}

Proof is based on \cite{zhu2017lectureMCMC}.
We define $x_i^{(-j)} := [x_i^{(1)}, \dots, x_i^{(j-1)}, x_i^{(j+1)}, \dots, x_i^{(d)}]^\top$. 
In Gibbs sampling, we have:
\begin{align*}
Q(x_i; x_{i-1}) = P^*(x_i^{(j)} | x_{i-1}^{(-j)}) \overset{(a)}{=} P^*(x_i^{(j)} | x_{i}^{(-j)}),
\end{align*}
where $(a)$ is because in Gibbs sampling, the dimensions of  a sample are updated conditioned on the updated values of other dimensions in the same sample. 

According to the above expression and Eq. (\ref{equation_Metropolis_Hastings_p_accept}), we have:
\begin{align*}
p_\text{accept} &= \min\Big(\frac{P^*(x_i)\, Q(x_{i-1}; x_i)}{P^*(x_{i-1})\, Q(x_i; x_{i-1})}, 1\Big) \\
&= \min\Big(\frac{P^*(x_i)\, P^*(x_{i-1}^{(j)} | x_{i-1}^{(-j)})}{P^*(x_{i-1})\, P^*(x_{i}^{(j)} | x_{i}^{(-j)})}, 1\Big).
\end{align*}

By marginalization, we have:
\begin{align*}
&P^*(x_i) = P^*(x_i^{(j)} | x_{i}^{(-j)})\, P^*(x_{i}^{(-j)}), \\
&P^*(x_{i-1}) = P^*(x_{i-1}^{(j)} | x_{i-1}^{(-j)})\, P^*(x_{i-1}^{(-j)}).
\end{align*}
Therefore:
\begin{align*}
&p_\text{accept} \\
&= \min\Big(\frac{P^*(x_i^{(j)} | x_{i}^{(-j)})\, P^*(x_{i}^{(-j)})\, P^*(x_{i-1}^{(j)} | x_{i-1}^{(-j)})}{P^*(x_{i-1}^{(j)} | x_{i-1}^{(-j)})\, P^*(x_{i-1}^{(-j)})\, P^*(x_{i}^{(j)} | x_{i}^{(-j)})}, 1\Big) \\
&= \min\Big(\frac{ P^*(x_{i}^{(-j)})}{P^*(x_{i-1}^{(-j)})}, 1\Big).
\end{align*}
In transition from $x_{i-1}$ to $x_i$ for the $j$-th dimension, only the $j$-th dimension change; therefore, excluding the $j$-th dimension, these two are equal:
\begin{align*}
x_{i-1}^{(j)} = x_{i}^{(j)}.
\end{align*}
Hence:
\begin{align*}
&p_\text{accept} = \min\Big(\frac{ P^*(x_{i}^{(-j)})}{P^*(x_{i}^{(-j)})}, 1\Big) = \min(1,1) = 1.
\end{align*}
Therefore, Gibbs sampling accepts the proposed sample with probability one. Q.E.D.

\bibliography{References}
\bibliographystyle{icml2016}

\end{document}